\documentclass[review,authoryear]{elsarticle}
\usepackage{amsmath}
\usepackage{amssymb}
\usepackage{graphicx}
\usepackage{natbib}
\usepackage{bm}
\usepackage{url}
\usepackage{amsthm}
\usepackage{thmtools}
\usepackage[margin=2cm,nohead]{geometry}
\usepackage{setspace}
\usepackage{paralist}
\usepackage{appendix}

\declaretheorem[style=definition]{definition}
\declaretheorem[style=plain]{theorem}
\newcommand{\R}[1]{} %does nothing: can't use the one below when line numbering off

\begin{document}
\title{Size change, shape change, and the growth space of a community}
\author{Matthew Spencer,\\School of Environmental Sciences, University of Liverpool, Liverpool, L69 3GP, UK.\\m.spencer@liverpool.ac.uk}
\maketitle


Colour figures: figure \ref{fig:perturbation}; figure \ref{fig:succession}; figure \ref{fig:heronisland}.

Appendices:

\ref{ap:invariance} Invariance under exponential growth and scaling invariance

\ref{ap:existing} Existing measures of rate of shape change

\ref{ap:derive} Deriving a measure of rate of shape change

\ref{ap:maxinvar} What kinds of functions are scaling invariant?

\ref{ap:estimate} Estimating size and shape change

\ref{ap:ce} Measuring colonization and extinction

\ref{ap:pa} Size and shape change in presence-absence data

R code and data used in the examples (\verb+succession.zip+). FOR REVIEW, AVAILABLE FOR DOWNLOAD FROM \url{http://www.liv.ac.uk/~matts/rateofsuccession.html}.

%referencing revisions
%\R{A1} a test sentence\\
%\R{A2} another test sentence\\

%\linenumbers
%\modulolinenumbers[1]

\section*{Abstract}

Measures of biodiversity change such as the Living Planet Index describe proportional change in the abundance of a typical species, which can be thought of as change in the size of a community. Here, I discuss the orthogonal concept of change in relative abundances, which I refer to as shape change. To be logically consistent, a measure of the rate of shape change should be scaling invariant (have the same value for all data with the same vector of proportional change over a given time interval), but existing measures do not have this property. I derive a new, scaling invariant measure. I show that this new measure and existing measures of biodiversity change such as the Living Planet Index describe different aspects of dynamics. I show that neither body size nor environmental variability need affect the rate of shape change. I extend the measure to deal with colonizations and extinctions, using the surreal number system. I give examples using data on hoverflies in a garden in Leicester, UK, and the higher plant community of Surtsey. I hypothesize that phylogenetically-restricted assemblages will show a higher proportion of size change than diverse communities.

\clearpage

\section{Introduction}

The Living Planet Index \citep{Loh05} and similar indices \citep[e.g.][]{Buckland11, Burns13, Dirzo14} measure proportional change in the abundance of a typical species. They are used by a number of major conservation organizations and have resulted in alarming headline figures such as the mean $25\%$ decline in abundance of terrestrial vertebrate populations between 1970 and 2000 \citep{Loh05}. I will refer to the property measured by these indices as ``size'' change. It is obvious that there are other aspects of change in the abundances of species in a community, which may also be of both theoretical and applied interest. I will concentrate on one of these:
\begin{definition}
``Shape'' change is ``change in the relative abundances of species in a community'' \citep{Lewis78}, where the relative abundance of a species is its abundance divided by the sum of abundances of all species in the community. 
\label{def:ratesuc}
\end{definition}
\citet{Lewis78} used Definition \ref{def:ratesuc} as part of a definition of succession, and measuring the rate of change in relative abundances is an important aspect of testing theories about succession \citep[e.g.][]{Boit14, Walker14}. However, in this manuscript I avoid the term ``succession'', because it has been associated with a much wider set of changes in communities and ecosystems \citep{Odum_767}. A number of measures of the rate of shape change have been proposed \citep{Jassby74, Lewis78, Foster00}, and measures developed for other purposes \citep{Field82, Legendre01} have also been used to measure shape change. Nevertheless, there has been little systematic consideration of how such a measure should behave, or of the connection to size change. \R{A56}I believe that the field would benefit from the kind of systematic approach that has been applied to measures of evenness \citep[e.g.][]{Smith96} and other aspects of diversity \citep[e.g.][]{Jost07,Leinster12}.

In this paper, I first outline the properties that a measure of rate of shape change should have. This \R{A26}gives a rough idea of how a suitable measure can be derived. After deriving such a measure, I show that it has two simple interpretations: first, as proportional to the Aitchison distance \citep{Aitchison92} between two sets of relative abundance data; and second, as the among-species sample standard deviation of mean proportional population growth rates over a finite time period. Note that the proportional growth rate for a species with positive abundance $x(t)$ at time $t$ is defined as $(1/x) \mathrm{d} x / \mathrm{d} t$, assuming that $x(t)$ is differentiable with respect to $t$\R{A13}\footnote{Widely-used near-synonyms include `per capita growth rate' and `specific growth rate'. However, `per capita' is not strictly appropriate when individuals are not well defined (e.g. clones, colonies) or not measured (e.g. proportional cover, biomass), and the term `proportional growth rate' agrees with the usages `proportional change' and `proportional scale', which are important in this context.}. I define the growth space of a community as a real space whose axes represent proportional growth rates of each species. \R{A70}I show that the new measure of rate of shape change and the Living Planet Index \citep{Loh05} are proportional to the lengths of projections onto orthogonal subspaces of growth space, and thus measure distinct aspects of community dynamics. \R{A72}The geometry of growth space leads to the results that body size, generation time and environmental variability have no necessary connection to the rate of shape change. I illustrate the calculation of the rate of shape change using data on hoverflies in a garden in Leicester, UK. I then extend the approach to deal with colonization and extinction, making use of the surreal number system, \R{A30}which includes quantities further from and closer to zero than any real number \citep{Conway01}. The surreal numbers are useful because, for example, an extinction should represent a larger change than any reduction in abundance that does not involve extinction, and yet such reductions can lead to arbitrarily large real numbers. I apply this extended method to the higher plant community on the volcanic island of Surtsey. Finally, I discuss the hypothesis that phylogenetically-restricted assemblages will show a relatively high proportion of size change compared to diverse communities. \R{A80}Notation used throughout is summarized in Table \ref{tab:symbols}.

\begin{table}
\caption{Notation. Major symbols used in both main text and appendices. Quantities for which no dimensions are indicated are dimensionless. Abundances are shown as dimensionless here, although in some cases they will have dimensions such as numbers per unit area.}
\label{tab:symbols}
\begin{tabular}{ l l l}
Symbol & Meaning & Dimensions\\
\hline
$c_i$ & observed count of the $i$th species\\
$k_1$ & number of species present at start and end of a time interval\\
$k_2$ & number of species present at start but absent at end of a time interval\\
$k_3$ & number of species absent at start but present at end of a time interval\\
$k_4$ & number of species absent at start and end of a time interval\\
$\mathbf M$ & an $ n \times n$ diagonal matrix with positive diagonal elements $m_i$\\
$n$ & number of species \\
$p_i$ & relative abundance $x_i/(\sum_{j=1}^n x_j)$ of the $i$th species \\
$\mathbf q$ & projection of $\mathbf r$ onto the subspace orthogonal to line of equal proportional growth rates & $T^{-1}$\\
$r$ & proportional growth rate $(1/x) \mathrm{d} x / \mathrm{d} t$ & $T^{-1}$\\
$\tilde{r}_i(t, t + \Delta t)$ & mean proportional growth rate over time interval $(t, t + \Delta t]$ & $T^{-1}$\\
$\mathbf r$ & vector of proportional growth rates (or mean proportional growth rates) & $T^{-1}$\\
$\bar{r}$ & among-species sample mean of mean proportional growth rates & $T^{-1}$\\
$s_r$ & among-species sample standard deviation of mean proportional growth rates & $T^{-1}$\\
$\mathcal S_1$ & set of species present at start and end of a time interval\\
$\mathcal S_2$ & set of species present at start but absent at end of a time interval\\
$\mathcal S_3$ & set of species absent at start but present at end of a time interval\\
$\mathcal S_4$ & set of species absent at start and end of a time interval\\
$t$ & time & $T$\\
$\Delta t$ & a time interval (not necessarily small) & $T$\\
$\theta$ & angle between $\mathbf r$ and line of equal proportional growth rates\\
$\mathbf u$ & projection of $\mathbf r$ onto the line of equal proportional growth rates & $T^{-1}$\\
$v_{ij}$ & ratio of abundances $x_i/x_j$ of a pair of species at a given time\\
$\mathbf v$ & vector of abundance ratios $v_{ij}$\\
$w_k$ & mean proportional rate of change in the $k$th abundance ratio in $\mathbf v$\\
$\mathbf w$ & vector of mean proportional rates of change in abundance ratios in $\mathbf v$\\
$x_i$ & abundance of the $i$th species\\
$\mathbf x, \mathbf y$ & abundance vectors for a set of species\\
$\psi$ & the surreal number $\omega^{1/\omega}$\\
$\omega$ & simplest surreal number greater than all positive real numbers\\
\end{tabular}
\end{table}

\section{Properties of a measure of rate of shape change}

In this section, I list the properties that I believe a measure of the rate of shape change should have in order to match up with biological intuition. \R{A27}If we accept these properties, we will know a lot about what a suitable measure should look like, even before we have attempted to derive it. \R{A64}It is of course true that others might choose a different list of properties, and thus arrive at measures that will be useful for different purposes.

\begin{enumerate}[Property 1.]

\item \label{property:ratefunc} By Definition \ref{def:ratesuc}, the rate of shape change should be expressible as a function of relative abundances and time alone. Although not explicit in the definition, I will here be concerned with scalar measures of rate (in physical terms, ``speed'' rather than ``velocity''). This is needed because one will often want to make comparisons of rates between communities containing different sets of species, which cannot be done if the rate has multiple components, each associated with a particular species or group of species.

If the community consists of $n$ species (not all of which may be observed at a particular time), then denote by $x_i \geq 0$ a measure of the abundance of the $i$th species, $i=1,2,\ldots,n$. The abundances of all species in the community at time $t$ can then be represented as a column vector $\mathbf x(t) \in \mathbb R_{\geq 0}^n$ (the $n$-dimensional space of positive real numbers).

The relative abundance of the $i$th species at time $t$ is $p_i(t) = x_i(t)/\sum_{j=1}^n x_j(t)$. The set of relative abundances for all species is also known as a composition, and is a vector in the unit simplex $\mathbb S^{n-1}$. Although it is necessary that a rate of shape change can be expressed in terms of relative abundances and time alone, many other necessary properties are easier to interpret in terms of absolute abundances.

\item \label{property:timefunc} The rate should be a mean rate of change over a finite time interval. This is required if such a rate is to be calculated from observations taken at discrete points in time. Furthermore, it is often useful to average out stochastic variability in community composition, so as to focus on underlying trends. 

Given observations at the start and end of a time interval $(t, t+\Delta t]$ and the intuitive concept of speed, the rate of shape change should be a function of the form
\begin{equation*}
f(\mathbf x (t), \mathbf x(t + \Delta t), \Delta t) = \frac{1}{\Delta t} g(\mathbf x (t), \mathbf x(t + \Delta t)).
\end{equation*}
This means that we need only find a suitable way of measuring the difference between two abundance vectors.

\item \label{property:scaleinvar} Defining shape change as change in relative rather than absolute abundances (Definition \ref{def:ratesuc}) means that it should not be altered if all abundances at a given time are multiplied by a constant. In other words,
\begin{equation}
\begin{aligned}
g(\mathbf x,\mathbf y) = g(\phi \mathbf x, \rho \mathbf y), \quad \phi>0, \rho>0.
\end{aligned}
\label{eq:scaleinvar}
\end{equation}
\R{A34}Any function of a composition satisfying Equation \ref{eq:scaleinvar} is expressible in terms of ratios of the form $x_i/x_j$ \citep{Aitchison92}.  This property also means that we can work with samples from which absolute abundance data are not available (which can apply to everything from light traps to environmental sequencing).

\item \label{property:zero} If no net change has occurred over a time interval, the rate of shape change over that interval should be zero. Intuitively, if all the abundances are unchanged, a measure that depends only on those abundances and the time interval should also be unchanged. In other words,
\begin{equation}
\begin{aligned}
g(\mathbf x,\mathbf x) = 0.
\end{aligned}
\label{eq:nochange}
\end{equation}

\item \label{property:nonneg} The rate of shape change should be non-negative. Intuitively, the community is different in some way if any of the relative abundances have changed, and a measure of the ``speed'' of change over some time interval should be positive if any change has occurred. In other words,
\begin{equation*}
\begin{aligned}
g(\mathbf x,\mathbf y) \geq 0.
\end{aligned}
\end{equation*}
\R{A79}Properties \ref{property:zero} and \ref{property:nonneg} together imply that the rate of shape change cannot be additive over subintervals. Suppose that over two successive subintervals we move from $\mathbf x$ to $\mathbf y \neq \mathbf x$ and then back to $\mathbf x$. From property \ref{property:zero}, the net change over the combined interval is $g(\mathbf x, \mathbf x) = 0$. From property \ref{property:nonneg}, the net change over each of the subintervals is positive.  

\item \label{property:neutral} If all species have the same proportional growth rate, relative abundances do not change, and the rate of shape change should be zero. This is the important special case of a \R{A2} deterministic neutral community. If Equations \ref{eq:scaleinvar} and \ref{eq:nochange} hold, setting $\mathbf y = \mathbf x$ in Equation 1 shows that this will be the case. In all other cases, the rate should not be zero. In other words,
\R{A1} \begin{equation*}
\begin{aligned}
g(\mathbf x,\mathbf y) = 0 \text{ if and only if } \mathbf y = \phi \mathbf x, \, (\phi>0).
\end{aligned}
\end{equation*}

\item \label{property:perturbinvar} \R{A76}The rate of shape change should be the same for all abundance trajectories having the same vector of proportional changes over some specified time interval. This is both the most important and the most difficult requirement to grasp, and is needed in order to make the rate of shape change consistent with the way in which population growth is measured.

For a single species, the same amount of proportional growth has occurred if the population of that species increases from 1 to 10 as if it increases from 10 to 100 units. This can be justified at the individual level by considering the population as a birth-death process. If the population changes by the same proportion in each case, then any individual experiences the same balance between proportional birth and death rates. Similarly, for a set of three species, the amount of proportional change from $[1,1,1]^{\prime}$ (where the prime denotes transpose) to $[1.1,0.9,1.2]^{\prime}$ is the same as from $[10,10,10]^{\prime}$ to $[11,9,12]^{\prime}$, and from $[1,10,1]^{\prime}$ to $[1.1,9,1.2]^{\prime}$. Thus, we require,
 \begin{equation}
g(\mathbf x, \mathbf y) = g(\mathbf M \mathbf x, \mathbf M \mathbf y)
\label{eq:perturbinvar}
\end{equation}
for all $n \times n$ matrices $\mathbf M$ with positive diagonal elements $m_i$ and zero off-diagonal elements. In other words, there is an equivalence class consisting of all pairs $(\mathbf x, \mathbf y)$ such that for all $i \in 1, \ldots, n$, $y_i / x_i = m_i$ for some species-specific constants $m_i$. \R{A33}Since multiplication by $\mathbf M$ corresponds to a scaling transformation with factors $m_1, m_2, \ldots m_n$, a function satisfying Equation \ref{eq:perturbinvar} may be called scaling invariant\footnote{The term used in compositional data analysis is perturbation invariant \citep[e.g.][]{Aitchison92, Egozcue03}, but in ecology this is likely to cause confusion with the idea of a perturbation as a disturbance to a system. Note that scaling invariance is distinct from Property \ref{property:scaleinvar}, which is sometimes called scale invariance.}. \R{A35}Figure \ref{fig:perturbation} explains scaling invariance visually. Consider any arbitrary abundance trajectory (left, dashed lines, here shown for three species), with two abundance vectors $\mathbf x$ and $\mathbf y$ lying on this trajectory, separated by time $\Delta t$. Applying the scaling $\mathbf M$ to the entire trajectory transforms $\mathbf x$ and $\mathbf y$ to $\mathbf M \mathbf x$ and $\mathbf M \mathbf y$ respectively (right). A function is scaling invariant if its value is unchanged by the action of any scaling matrix $\mathbf M$ with positive diagonal elements on any two abundance vectors $\mathbf x$ and $\mathbf y$.

\begin{figure}
\includegraphics[height=14cm]{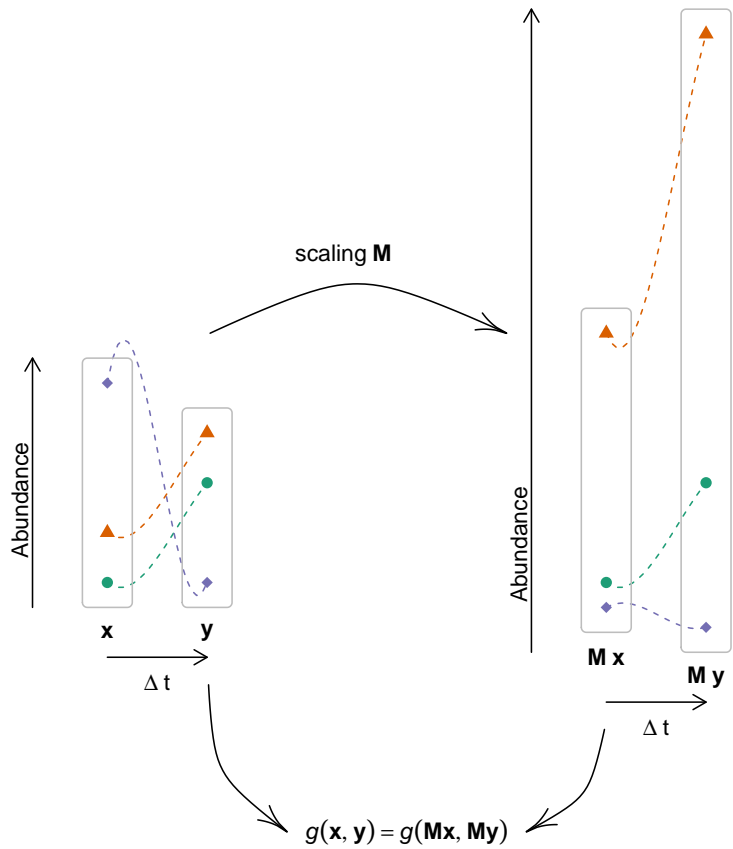}
\caption{Scaling invariance (Property \ref{property:perturbinvar}). Left: arbitrary abundance trajectory for three species represented as dashed lines, with abundance on the vertical axis and time on the horizontal axis. Abundance vectors $\mathbf x$ and $\mathbf y$ lie on this trajectory, separated by time $\Delta t$ (circle, triangle and diamond symbols represent the three abundances in each vector). Applying the scaling $\mathbf M$ to the whole trajectory leads to corresponding abundance vectors $\mathbf M \mathbf x$ and $\mathbf M \mathbf y$ (right). If the function $g(\mathbf x, \mathbf y)$ is scaling invariant, then it is unchanged by the action of the scaling $\mathbf M$.}
\label{fig:perturbation}
\end{figure}

\R{A37}Scaling invariance is important theoretically if we want to take an organism-centred view of temporal dynamics. For any species, level sets in niche space are sets of points for which the species has the same proportional growth rate, and therefore experiences an environment of equivalent quality. This idea extends the concept of the Hutchinson niche, which is bounded by the level set of zero proportional growth rate \citep{Maguire73}. A scaling $\mathbf M$ applied to an abundance trajectory does not alter the proportional growth rate for any species at any time, and therefore should not alter a measure that tells us about how organisms experience change.

An unrelated but important practical consequence of scaling invariance is that there is no need to measure abundance in the same units for every species, provided that the same units are used for a given species at all time points. For example, the abundance of one species could be measured in percentage cover, and that of another species in numbers trapped in a pitfall trap, without affecting the value of a scaling invariant measure of rate of change. In this case the matrix $\mathbf M$ can be thought of as representing the conversion factors needed to put all abundances in the same units, and if Equation \ref{eq:perturbinvar} holds, $\mathbf M$ does not need to be known. \R{A3}Similarly, if detection probabilities differ among species but remain constant over time, they do not affect the value of a scaling invariant measure of rate of change (although if detection probabilities change over time in different ways for different species, such a measure will be affected).

There is one obvious objection to the idea that change should be measured on a proportional scale: it might be argued that the proportional scale gives too much weight to rare species. Many important properties of communities and ecosystems can be approximated as functions of abundance. Examples include the structural complexity of the community, the rate of primary production, and the conservation value of the community. The simplest plausible functions relating these properties to abundance are linear combinations of abundances, where each species has a weight that measures the contribution of a unit of abundance of that species to the function \citep[e.g.][]{Gross05}. If the abundance of a rare species is doubled, many of these properties might be almost unchanged. In contrast, doubling the abundance \R{A4}of a common species might result in large changes in these properties. 

There are two answers to this objection. \R{A48}First, a scaling invariant function must give equal weight to every species (\ref{ap:maxinvar}). Thus, the decision to take an organism-centred view of temporal dynamics (which leads to scaling invariance) precludes also taking the ecosystem function-centred view outlined above. Second, I believe that if changes in ecosystem function are of interest, they should be studied directly, rather than through using change in abundances as a surrogate. Rare species will sometimes be important to ecosystem function. For example, a rare but very large tree species might have a big effect on structural complexity, a rare keystone predator might have a big indirect effect on primary production, and a rare but endangered organism such as a rhinoceros might have a big effect on conservation value. Thus, it is meaningless to talk about a measure that gives ``too much weight'' to a rare species unless the right weight to give each species is known. If the weights are known, the appropriate function can be studied. If the weights are unknown, studying abundances will be misleading. Even the simplest families of population growth models such as the exponential and the logistic are not closed under addition \citep[p. 89]{Kingsland85}. Thus, a weighted sum of abundances of populations, each growing according to a different member of such a family, will have qualitatively different behaviour from any member of the family. In summary, changes in abundance should be studied only if they are of interest in their own right, not as a surrogate for some other property.

\R{A44}Exponential growth is an important idealized case in population dynamics, playing a similar role in population biology to that of a body with no forces acting on it in Newtonian physics \citep[chapter 6]{Ginzburg04}. It also provides an easy way to demonstrate whether a measure of the rate of shape change could be scaling invariant. In this special case, each population experiences an environment of constant quality, and grows at a constant proportional rate. In other words all relevant aspects of the system have either zero or non-zero but constant proportional rates of change. A necessary but not sufficient condition for scaling invariance is that a measure of the rate of shape change is constant over time when every species is growing exponentially (\ref{ap:invariance}). \R{A28}The measures of rate of shape change in common use \citep{Jassby74, Lewis78, Field82, Foster00, Legendre01} are not constant over time under exponential growth (Figure \ref{fig:succession}A and \ref{ap:existing}), and are therefore not scaling invariant. For example, it is easy to see that measures of the rate of shape change based on squared \citep{Foster00} or absolute \citep{Lewis78} differences in composition will not be invariant under constant proportional changes in composition resulting from exponential growth.

\begin{figure}
\includegraphics[height=14cm]{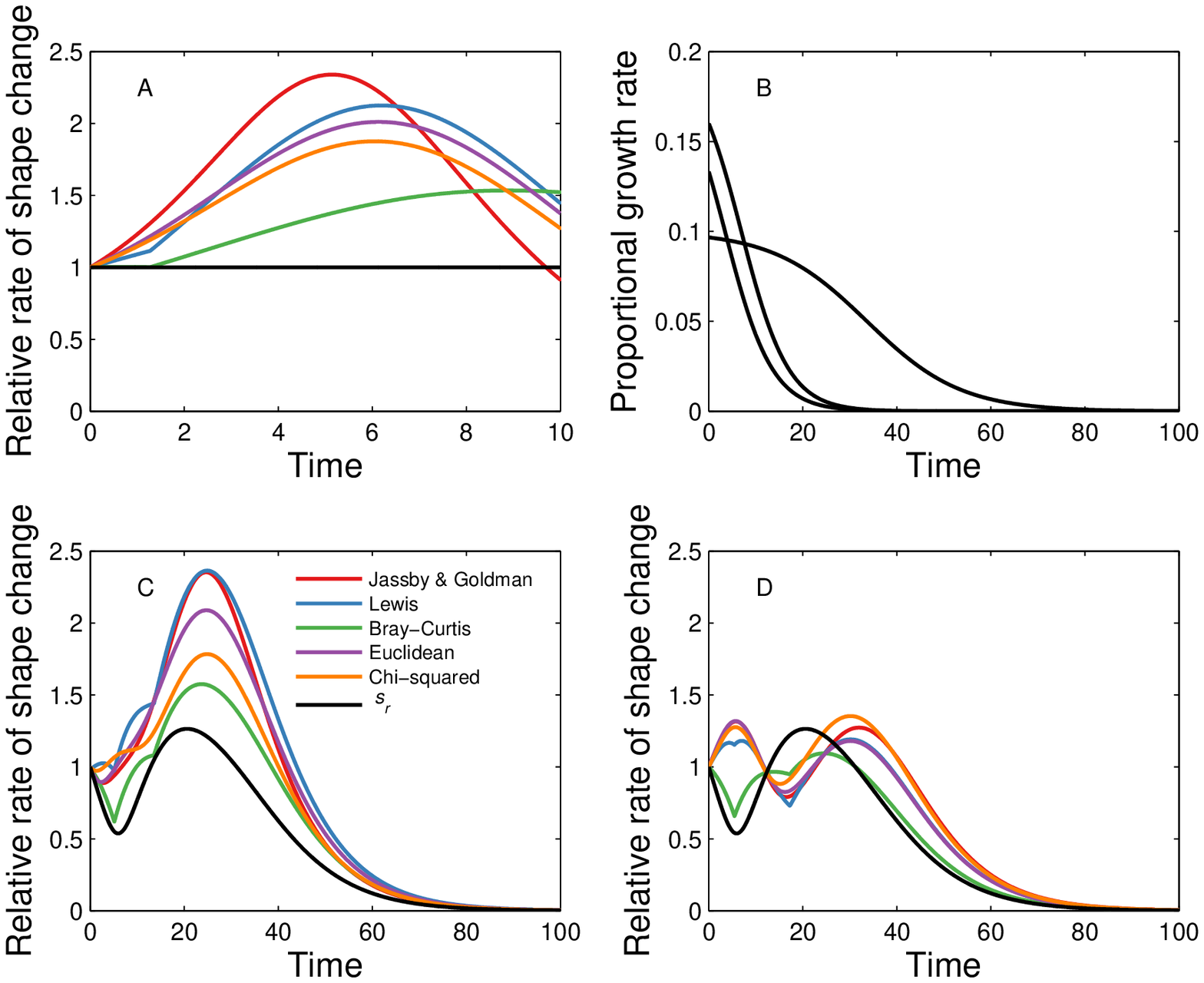}
\caption{A: relative rate of shape change against time for three species growing exponentially, as measured by $s_r$ (black line, Equation \ref{eq:sdgrowth}, the among-species standard deviation of proportional growth rates) and five existing measures of rate of shape change \citep[red, blue, green, purple, orange, respectively, as in panel C key]{Jassby74, Lewis78, Field82, Foster00, Legendre01}. B: proportional growth rate [$(1/x_i) \mathrm d x_i / \mathrm d t$] against time for three species growing logistically. C: relative rate of shape change against time for three species growing logistically as in panel B. D: relative rate of shape change against time for three species growing logistically, as in panel B except that initial abundances and carrying capacities (and all abundances on the trajectory between these points) are scaled by $\mathbf M = \text{diag}(1/3, 1/9, 1/2)$, keeping proportional growth rates constant. In A, C, and D: Euclidean is $1/\Delta t$ times the distance in \citet{Foster00}; Chi-squared is $1/\Delta t$ times the distance in \citet{Legendre01}; Bray-Curtis is $1/\Delta t$ times the distance in \citet{Field82}, calculated on relative abundances with a 1/4 power transform; constant time interval $\Delta t = 0.001$ between observations; each measure of rate of shape change is plotted relative to its value over the first time interval. Parameters for A: initial abundances $[0.5, 0.1, 0.4]^\prime$; proportional growth rates $[-0.2, 0.2, -0.1]^\prime$. Parameters for B and C: initial abundances $[1, 0.3, 0.4]^\prime$, proportional growth rates at low abundance $[0.2, 0.1, 0.2]^\prime$, carrying capacities $[3, 9, 2]^\prime$. Parameters for D: initial abundances $[1/3, 1/30, 1/5]^\prime$, proportional growth rates at low abundance $[0.2, 0.1, 0.2]^\prime$, carrying capacities $[1, 1, 1]^\prime$. }
\label{fig:succession}
\end{figure}

The arguments above do not need the assumption that exponential growth is common in nature. Any two patterns of growth resulting in the same pattern of proportional change over a given time interval should correspond to the same rate of shape change, no matter what trajectories the populations involved followed to get from the start to the end of the interval.

\item \label{property:richness} A measure of the rate of shape change should not depend (in a statistical sense) on the number of species being studied. This is important for two reasons. First, it is usually impossible to measure the abundance of every species in the community. If instead we have abundances for a random sample of species, then a measure whose expected value does not depend on the number of species sampled can be used to estimate the rate of shape change in the entire community. \R{A10}In general, it will be difficult to obtain a random sample of species. Selecting species at random from a complete species list is rarely possible. Choosing individuals at random and identifying them does not give a random sample of species, because more abundant species are more likely to be included in the sample. In practice, available data are often used without any attempt at random sampling \citep[e.g.][p. 72]{Collen13}. Nevertheless, independence of the number of species studied remains a desirable statistical property.

Second, we will often want to compare communities with different numbers of species. Other things being equal, there is no intuitive reason why a change in the number of species should be associated with a change in the rate of shape change.

\end{enumerate}

\section{A new measure of rate of shape change}

In what follows, I initially restrict attention to cases where there is no colonization or extinction, so that all abundances are positive. I later show how to deal with colonization and extinction in a way that is consistent with the approach taken below.

Consider a single species with abundance $x(t)$ at time $t$, where $x(t)$ is a \R{A11}differentiable, positive function of time, and proportional growth rate $(1/x) \mathrm{d}x / \mathrm{d} t$. The mean proportional growth rate over the time interval $(t,t + \Delta t]$ is
\begin{equation}
\tilde{r}(t,t+\Delta t) = \frac{1}{\Delta t} \int_t^{t+\Delta t} \frac{1}{x} \frac{\mathrm d x}{ \mathrm d u} \, \mathrm du = \frac{\log x(t + \Delta t) - \log x(t)}{\Delta t}, \quad x>0,
\label{eq:meanrate}
\end{equation}
(throughout, $\log$ denotes the natural logarithm). \R{A12}The requirement that $x(t)$ is differentiable is not much of a limitation, because even in stochastic models where abundance may not be differentiable, there is usually a differentiable expected abundance to which we can apply Equation \ref{eq:meanrate}. \R{A74}The mean proportional growth rate is obviously a constant for all abundance trajetories resulting in the same amount of proportional growth over a given time interval. I show in \ref{ap:derive} that the natural generalization of Equation \ref{eq:meanrate} to a scalar measure of the rate of shape change is the among-species sample standard deviation of mean proportional growth rates,\R{A19}
\begin{equation}
s_r = \left [\frac{1}{n-1} \sum_{i=1}^n (\tilde{r}_i (t, t + \Delta t) - \bar{r}(t+ \Delta t))^2 \right] ^{1/2},
\label{eq:sdgrowth}
\end{equation}
where $\tilde{r}_i(t, t + \Delta t)$ is the mean proportional growth rate of the $i$th species over the time interval $(t, t + \Delta t]$, and $\bar{r}(t, t + \Delta t)$ is the among-species mean of mean proportional growth rates over this interval:
\begin{equation}
\bar{r}(t, t + \Delta t) = \frac{1}{n} \sum_{i=1}^n \tilde{r}_i(t, t + \Delta t).
\label{eq:rbar}
\end{equation}
I show in \ref{ap:derive} that Equation \ref{eq:sdgrowth} satisfies Properties \ref{property:ratefunc} to \ref{property:richness} \R{A46}(although Property \ref{property:richness} is only satisfied asymptotically). In particular, this measure is scaling invariant, unlike existing measures of the rate of shape change. For example, the black line in Figure \ref{fig:succession}A is the value of $s_r$ for a numerical example in which all species are growing exponentially. \R{A23}I show in \ref{ap:derive} that Equation \ref{eq:sdgrowth} is proportional to the Aitchison distance between the relative abundances at two time points. The Aitchison distance is the standard measure of distance between compositions \citep{Aitchison92, Egozcue03}. \R{A77}Finally, I show in \ref{ap:derive} that Equation \ref{eq:sdgrowth} can also be expressed in terms of relative abundances:
\begin{equation}
s_r = \frac{1}{\Delta t}\left[ \frac{1}{n(n-1)} \sum_{i=1}^{n-1} \sum_{j=i+1}^n \left( \log\left( \frac{p_i(t + \Delta t)}{p_j (t+\Delta t)} \right) - \log\left( \frac{p_i(t)}{p_j(t)}  \right)  \right)^2 \right]^{1/2}.
\label{eq:srp}
\end{equation}

\R{A24}Although I have shown that $s_r$ is scaling invariant while a number of other proposed measures of shape change are not, one might wonder whether there are other possible scaling invariant measures of shape change. In \ref{ap:maxinvar}, I show that any scaling invariant function is a function of the vector $\mathbf r$ of mean proportional growth rates. Among such functions, a function proportional to the Aitchison distance is in some sense the simplest way to measure change in relative abundances \citep{Aitchison92}.

\section{A logistic growth example}

\R{A49}In this section, I illustrate how $s_r$ and existing measures of the rate of shape change behave in the simple case of logistic growth. \R{A52}Exponential growth is defined by constant proportional growth rates. Under patterns of growth other than exponential, proportional growth rates do not remain constant over time, and so the rate of shape change should not in general be constant. For example, Figure \ref{fig:succession}B shows proportional growth rates for three species under logistic growth, and Figure \ref{fig:succession}C shows the corresponding values of $s_r$ (black line) and several existing measures of rate of shape change. Applying any non-identity scaling to the abundance trajectories (for example, $\mathbf M = \text{diag}(1/3, 1.9, 1/2)$) leaves the proportional growth rates (Figure \ref{fig:succession}B) and $s_r$ (Figure \ref{fig:succession}D, black line) unchanged, but causes qualitative changes in existing measures of the rate of shape change (Figure \ref{fig:succession}D, red, blue, green, purple and orange lines). In Figure \ref{fig:succession}D, the Jassby and Goldman, Lewis, Bray-Curtis, Euclidean and chi-squared measures have local minima just before time 20, which are absent in Figure \ref{fig:succession}C. All these measures can also differ qualitatively from $s_r$. In Figure \ref{fig:succession}D, the Jassby and Goldman, Lewis, Euclidean and chi-squared measures have local maxima close to the time at which $s_r$ has its first local minimum, and the Bray-Curtis measure has a local minimum close to the time at which $s_r$ has its first local maximum. It is also obvious from Figure \ref{fig:succession}C and D that rates of shape change need not be decreasing functions of time, even under very simple models of population growth. It has been suggested that during succession, \R{A53}``rate-of-change curves are usually convex, with change occurring most rapidly at the beginning'' \citep{Odum_767} for a wide range of ecosystem properties. The results in Figure \ref{fig:succession} suggest that if real communities typically have decreasing rates of shape change over time, some biologically interesting mechanism \R{A14}must be generating them.\R{A29}
 
\section{Shape change and biodiversity trends}

In this section, I show the geometric connection between the new measure of rate of shape change (Equation \ref{eq:sdgrowth}) and three high-profile measures of biodiversity trends. This connection is a major strength of $s_r$, not shared by existing measures of the rate of shape change\R{A68}. I consider three major measures of biodiversity trends: the the Living Planet Index \citep{Loh05}, which measures global changes in vertebrate populations; the UK Wild Bird Indicators \citep{Buckland11}, which measures changes in UK bird populations; and the Watchlist Indicator, which measures population trends in 155 species of birds, mammals, butterflies, and moths of conservation priority in the United Kingdom \citep[pages 10 and 81]{Burns13}. All these measures are based on the exponential of the among-species mean of mean proportional growth rates. They reveal patterns of major conservation importance, such as the mean 25$\%$ decline in terrestrial vertebrate populations between 1970 and 2000 \citep{Loh05}, and the mean 77$\%$ decline in species of conservation priority in the United Kingdom between 1970 and 2010 \citep[p. 10]{Burns13}. 

Such declines do not necessarily involve shape change in the sense of Definition \ref{def:ratesuc}: if all species decline at the same rate, fewer organisms are present in total, but their relative abundances are unchanged. To understand the connection between shape change and these measures of biodiversity trends, I introduce the idea of the growth space of a community. For a given set of species, every vector of proportional growth rates can be represented as a point in a space $\mathbb R^n$ which I refer to as the growth space of the community (Figure \ref{fig:projection} is a two-dimensional example). In the following I drop the time indexing for simplicity. The line $r_1 = r_2 = \ldots = r_n$ in growth space represents a deterministic neutral community, in which all species have the same proportional growth rate and the rate of shape change is zero (Property \ref{property:neutral}). The growth rate vector $\mathbf r$ can be decomposed into two orthogonal components (\ref{ap:derive}). The first component $\mathbf u$ is the projection of $\mathbf r$ onto the line of equal proportional growth rates, which represents change in abundance without shape change. \R{A18}The among-species mean of mean proportional growth rates (Equation \ref{eq:rbar}) is proportional to the length of this component and is a natural measure of size change. Size change is therefore essentially the same as the natural log of the measures of biodiversity trends described above. This component is important because change purely in the $r_1=r_2=\ldots=r_n$ direction represents balanced exponential growth, a state disallowed by the standard Lotka-Volterra model of consumer-resource dynamics, but possible under ratio-dependent dynamics \citep[section 6.1]{Arditi12}.

\begin{figure}
\includegraphics[height=14cm]{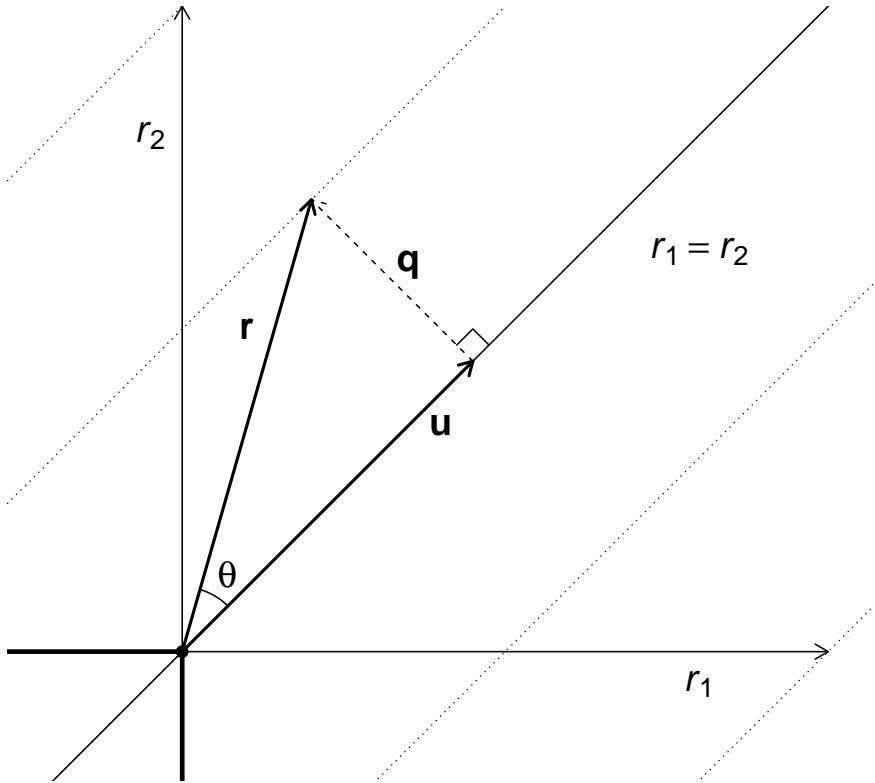}
\caption{Geometric interpretation of the rate of shape change in a two-species community. The axes $r_1$ and $r_2$ are the proportional growth rates of the two species. The position vector $\mathbf r$ represents the mean proportional growth rates for a given time interval. The diagonal line $r_1 = r_2$ represents equal proportional growth rates, and therefore no shape change. The vector $\mathbf u$ is the projection of $\mathbf r$ onto the $r_1=r_2$ line. It represents the rate of change of size of the community, and each of its elements equals the among-species mean of mean proportional growth rates. The norm of the projection $\mathbf q$ of $\mathbf r$ onto the subspace orthogonal to the $r_1=r_2$ line represents change in shape of the community, and is proportional to the among-species standard deviation of mean proportional growth rates. The square of the cosine of angle $\theta$ is change in size as a proportion of total change. All the internal equilibria are mapped to the origin (black dot). On the bold half-lines $(r_1<0, r_2=0)$ and $(r_1=0, r_2 < 0)$, one species declines while the other remains at constant abundance. The dotted lines are contours of constant rate of shape change, with the rate increasing away from the $r_1 = r_2$ line.}
\label{fig:projection}
\end{figure}

However, it is unlikely that precisely balanced growth will occur in a real community. The length of the second component $\mathbf q$ measures the extent to which growth is unbalanced. It represents change in relative abundance (shape) and is proportional to the among-species standard deviation of mean proportional growth rates (Equation \ref{eq:sdgrowth}). The orthogonality of these two components implies that the proposed measure of rate of shape change is distinct from measures of biodiversity trends such as the Living Planet Index \citep{Loh05}, the UK Wild Bird Indicators \citep{Buckland11}, and the Watchlist Indicator \citep{Burns13}.

Two further quantities can be obtained from Figure \ref{fig:projection}. First, the total amount of change in a community is the norm $\| \mathbf r \|$ of the vector $\mathbf r$. This norm is zero if the community is at equilibrium, and positive otherwise. I refer to it as the ``activity'' level of the community (with dimensions time$^{-1}$). To make comparisons among communities, the scaled activity $n^{-1/2} \| \mathbf r \|$ is useful because its expected value does not depend on the number of species sampled. Second, $\cos^2 \theta = (\| \mathbf u \| / \| \mathbf r \|)^2$ (Figure \ref{fig:projection}) measures the proportion of change that is size change. The value of $\cos^2 \theta$ is a useful way of comparing the dynamics of communities because it is dimensionless, and does not depend on the number of species in the community. In terms of linear models, $\cos^2 \theta$ is the coefficient of determination for a one-dimensional linear model containing only a mean term \citep[p. 410]{Saville91}.

\section{Properties of growth space}

The geometry of growth space makes it easy to visualize some properties of the rate of shape change (Figure \ref{fig:projection}), and leads to important theoretical results. As already mentioned, the rate of shape change is zero along the line $r_1=r_2 = \ldots = r_n$ (Figure \ref{fig:projection}, solid line). Thus, the community can increase or decrease in total abundance without shape change, by remaining on this line. Lines parallel to the $r_1=r_2 = \ldots = r_n$ line (Figure \ref{fig:projection}, dotted lines) have constant rate of shape change, increasing away from the $r_1=r_2 = \ldots = r_n$ line in each direction. At the origin ($r_1=r_2 = \ldots = r_n = 0$, Figure \ref{fig:projection}, black dot), corresponding to all the equilibria of the system, there is no shape change and no growth. The bold half-lines in Figure \ref{fig:projection} (excluding the origin) consist of points for which some proportional growth rates are zero and others are negative. If all species had positive abundance for such a community, and growth rates remained at a fixed point on the bold half-lines, the community would asymptotically approach a boundary equilibrium at which the species with negative growth rates were absent. Since such points do not lie on the $r_1 = r_2 = \ldots = r_n$ line, the rate of shape change would not go to zero as the equilibrium was approached (although such situations are unlikely to be common, and species with very low abundance would eventually go extinct, as discussed later).

\R{A47}Similar but more complicated cases may also occur. In a Lotka-Volterra competition model for proportional cover of six components in a coral reef system \citep{ST08}, the abundance of macroalgae and pocilloporid corals is very close to zero from around 25 years onwards (Figure \ref{fig:heronisland}A, green and orange lines respectively). The proportional growth rates of these two components have large but fluctuating negative values (Figure \ref{fig:heronisland}B, green and orange lines respectively), while the proportional growth rates for the other components are much closer to zero. In growth space, the system is approximately moving relatively far from zero in the $(-,-)$ quadrant of a plane on which all components other than macroalgae and pocilloporid corals have zero proportional growth rates. Thus, from around 25 years onwards, the rate of shape change as measured by $s_r$ is large, but fluctuates over time (Figure \ref{fig:heronisland}C, black line). Existing measures of the rate of shape do not have large values for this portion of the model output (Figure \ref{fig:heronisland}C, red, blue, green, purple and orange lines). They behave differently from $s_r$ because they do not measure distance in growth space.

\begin{figure}
\includegraphics[height=14cm]{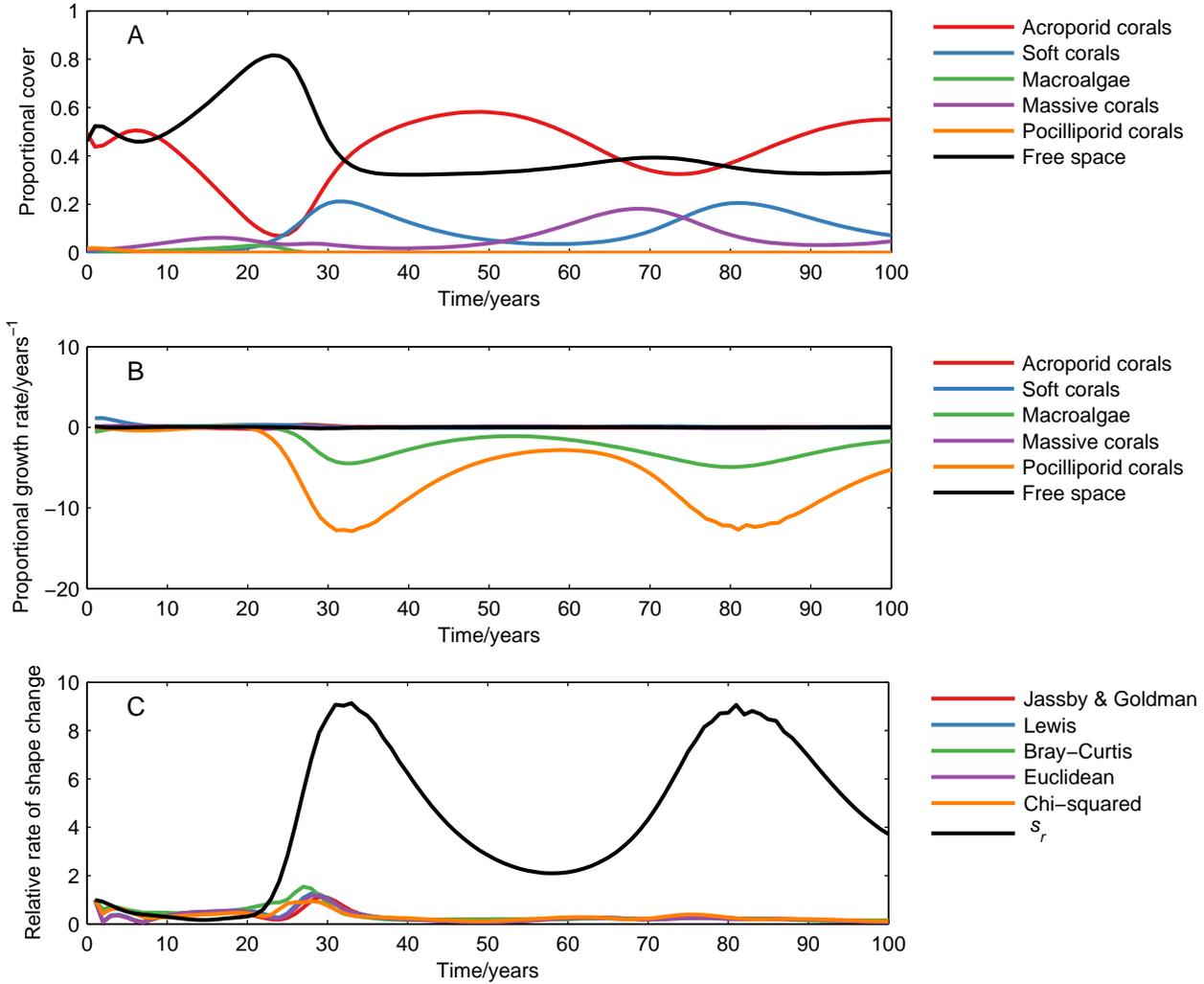}
\caption{A: changes in proportional cover over time in a 6-component Lotka-Volterra competition model for a coral reef system (Heron Island, Protected Crest site). Parameter estimates from Table 2 in \citet{ST08}. B: Mean proportional growth rates (over one-year intervals) against time for the data in A. C: $s_r$ (black line, Equation \ref{eq:sdgrowth}) and five existing measures of rate of shape change \citep[red, blue, green, purple, orange, respectively: ][]{Jassby74, Lewis78, Field82, Foster00, Legendre01}. Each measure is plotted relative to its value in the first time interval.}
\label{fig:heronisland}
\end{figure}

Two important theoretical results follow immediately from the geometry of growth space.\R{A62} First, there is no necessary connection between the typical body size and generation time for species in a community and the rate of shape change. Differences in generation time are likely to be associated with differences in body size and the among-species mean of maximum proportional growth rates \citep[e.g.][p. 230]{Drury73, Tilman88}. It is possible (although not necessarily true) that there are corresponding effects on mean proportional growth rates in the field. Nevertheless, it is only the among-species standard deviation of mean proportional growth rates that determines how fast the relative abundances change. In consequence, there is no reason to assume that shape change \R{A15}must be faster in, for example, a microbial community than in a forest. If shape change is generally faster in microbial than plant communities \citep{Fierer12}, it is not simply because microbes have faster growth rates. Similarly, although there may be a relationship between body size and the rate of shape change, as \citet{Anderson07} suggested in the broader context of succession, it is not a mathematical necessity.

Second, environmental variability \emph{per se} need not affect the rate of shape change. Although such variability may result in movement in growth space, the rate of shape change is only affected if this movement is not parallel to the $r_1=r_2 = \ldots = r_n$ line. In other words, environmental variability only affects the rate of shape change if it affects different species in different ways.\R{A69}

\section{Example: temporal change in a hoverfly assemblage, Leicester, UK}

The data analyzed here are a 30-year record of abundances of hoverflies (Diptera: Syrphidae) in a suburban garden in Leicester, United Kingdom. A Malaise trap was used to catch flying insects from 1 April to 31 October, on the same site every year from 1972 to 2001 \citep[p. 37]{Owen10}. Hoverflies were studied in more detail than other insect groups, and annual numbers for the 14 species caught more than 1000 times in total were reported \citep[p. 88]{Owen10}. Hoverflies in the UK have one, two, or  more than two generations a year \citep[p. 91]{Owen10}, so 30 years is a relatively long time scale for this assemblage. Larvae of different species use different food resources, including plant tissues, decaying organic matter, and aphids \citep[p. 80]{Owen10}, so the assemblage is quite functionally diverse. Large, active species are thought to be trapped less efficiently than small species \citep[p. 84]{Owen10}. \R{A7}Therefore, existing measures of rate of shape change, which are not scaling invariant (Property \ref{property:perturbinvar}), will not be useful descriptions of these data.

\R{A8}Variation in annual counts is a combination of deterministic trends in true abundance, stochastic variation in true abundance (process error) and observation error. I used a state-space model (\ref{ap:estimate}) to describe both process and observation error, and calculated rates of change in size and shape using the expected abundances from this model. The four occasions on which a zero count was recorded can plausibly be treated as observation error \R{A40}based on a Poisson sampling model: a species which was not recorded in the trap sample in a given year is still likely to have been present in the local area.

In Figure \ref{fig:hoverflies}A, the logs of observed counts $c_i+1$ are plotted against year, with one line for each species \R{A41}(note that the $\log (c_i + 1)$ transformation is used only for plotting the data: a Poisson sampling model is used for analysis). The slope of the line segment connecting the abundances of a given species in two successive years approximates the mean proportional growth rate $\tilde{r}_i$ for that species (except in the rare cases where the count was zero). Three important features are immediately obvious. First, there was a general downward trend in abundances, which may be caused by urbanization of surrounding agricultural land and a gradual increase in temperature \citep[p. 229]{Owen10}. Second, there were year-to-year fluctuations of more than two orders of magnitude. Third, these fluctuations often involved simultaneous increases or decreases in many species. Important causes of these fluctuations include warm weather in spring and summer, variation in food supply for species whose larvae feed on aphids, and immigration from surrounding agricultural land \citep[pp. 86-87]{Owen10}.

\begin{figure}
\includegraphics[height=18cm]{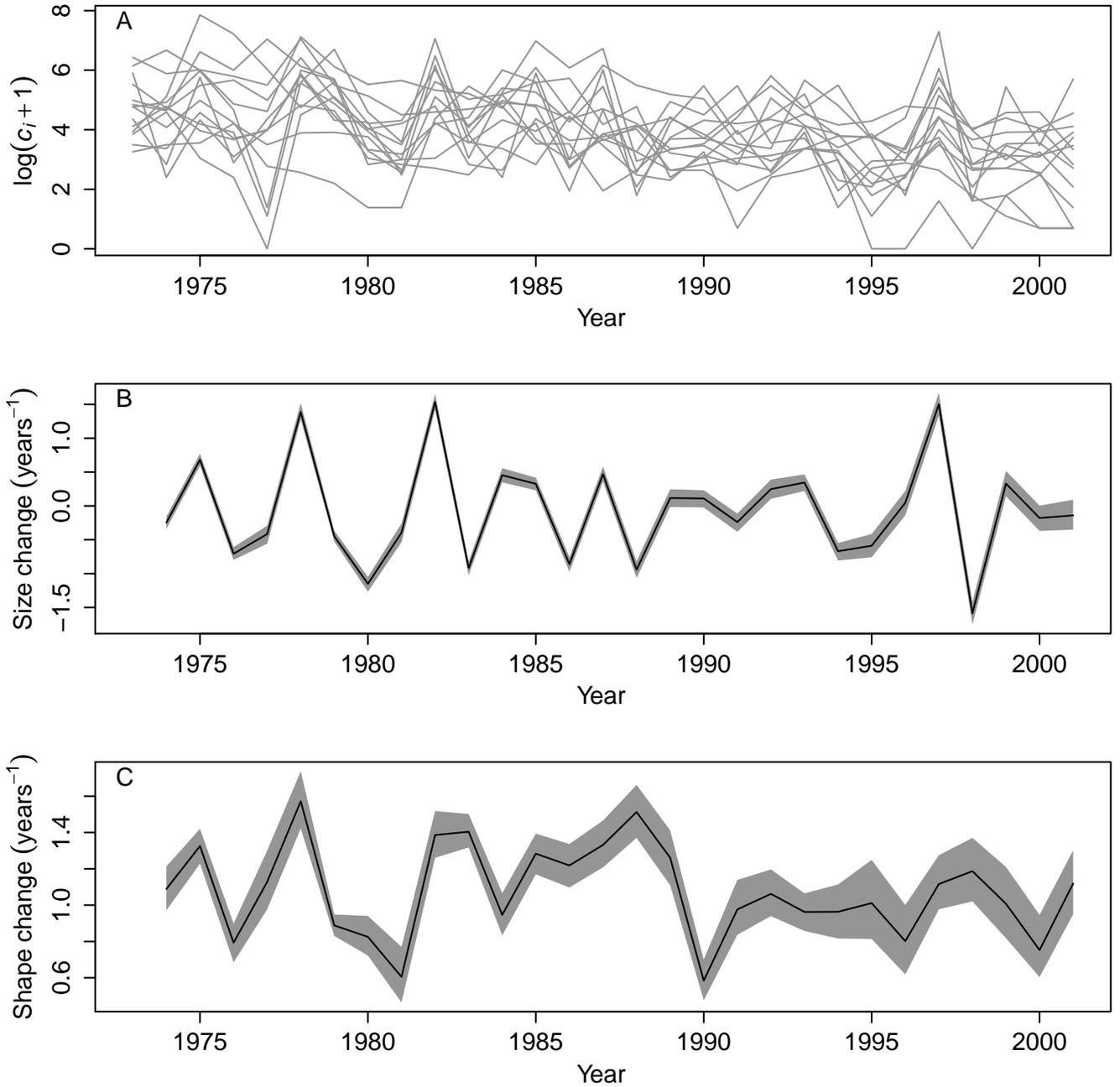}
\caption{Dynamics of the hoverfly assemblage of a Leicester, UK garden. A: log of observed count $c_i+1$ against time for the 14 species that were trapped more than 1000 times (one line per species). B: rate of change in size (among species mean of mean proportional growth rates $\bar{r}$, years$^{-1}$) against time. C: rate of change in shape (among-species standard deviation of mean proportional growth rates, $s_r$, years$^{-1}$) against time. In B and C, the black line is the mean estimate, and the grey area is a 95$\%$ confidence band.}
\label{fig:hoverflies}
\end{figure}

The size component (Figure \ref{fig:hoverflies}B) fluctuated a lot despite the general negative trend in abundances, and there were many intervals in which it was positive (indicating that a typical species increased in abundance over that period). \R{A9}Note that for most of these intervals, the 95$\%$ confidence band for size change did not include zero, so that these short-term fluctuations are likely to be real, rather than a consequence of observation error. Similarly, the most obvious feature of the shape component was its variability from year to year (Figure \ref{fig:hoverflies}C). \R{A59}I will return to the ecological meaning that can be extracted from these patterns after developing methods to deal with colonization and extinction.

\section{Colonization and extinction}

So far, I assumed that abundances are strictly positive (or at least that the underlying expected abundances are positive, if sampling zeros occurred). Even in cases where a complete census has been conducted, one could think of the set of individuals present as being a single realization of a stochastic process with positive expected abundances. \R{A67}Thus, treating absences as sampling zeros will usually be appropriate. Nevertheless, true zeros may occur, with a transition from a true zero to a positive abundance representing colonization and a transition from positive abundance to true zero representing extinction. Perhaps more importantly, the following analysis provides a link to related measures for presence/absence data.

If abundance $x_i$ of some species is zero, the proportional growth rate $(1/x_i) \mathrm{d} x_i / \mathrm{d} t$ is undefined. This has important practical consequences. For example, the Watchlist Indicator \citep{Burns13} is closely related to the rate of change in size (Equation \ref{eq:rbar}) for species of conservation priority. It would be undefined if one or more of the species of interest went extinct, yet such an event would be of even more concern than a decline in abundance without extinction.

Looking at the form of the mean proportional growth rate (Equation \ref{eq:meanrate}) suggests an approach. In the following explanation, I focus on extinction, because colonization is simply the mirror image. Intuitively, suppose that we would like an index based on a weighted sum of two terms, one reflecting changes in abundance measured on a proportional scale, and the other reflecting extinctions. The mean proportional growth rate (Equation \ref{eq:meanrate}) involves $\log x_i(t)$, and \R{A38}$\lim_{x_i(t+\Delta t) \rightarrow 0^+} \log x_i(t+\Delta t) = - \infty$. \R{A39}Choosing a real weight for extinctions is equivalent to adding an arbitrary constant to an observed zero, and it is in principle possible for abundance to be positive but less than this constant. Thus for every real choice of weight given to the extinction component and given initial abundance $x_i(t)$, there is a positive abundance $x_i(t+\Delta t)$ for which the magnitude of the mean proportional growth rate is greater than this weight. In other words, for any finite weight given to extinctions, there will be changes in abundance not involving extinction that count for more than changes involving extinction, which is intuitively unreasonable. Therefore the weight given to an extinction should have greater magnitude than any real number.

Simply calling the weight given to an extinction infinite is not satisfactory. Intuitively, one would like the extinction of a common species to represent a larger change than the extinction of a rare species. Also, one would like the transition from a given non-zero abundance to zero abundance to represent a more negative rate of change if it occurs over a shorter time interval $\Delta t$. Arithmetic operations involving infinite quantities are not defined in the real number field, so these properties will not hold.

There is a a natural solution. Instead of replacing zero abundances by $\lim_{x_i \rightarrow 0^+}$, replace them by a number larger than zero but smaller than all positive real numbers. Such numbers exist in the surreal number field \citep{Conway01}, and their logarithms have magnitude greater than all positive real numbers. The surreal number field contains the real numbers, but has arithmetic operations defined for infinite as well as finite numbers. These operations behave as expected when applied to real numbers. If a species has nonzero abundance at time $t$ but zero abundance at time $t+\Delta t$, define its mean proportional growth rate as $(\log (1/\omega) - \log x_i(t))/\Delta t$, where $1/\omega$ is the simplest number that is greater than zero but smaller than all positive real numbers \citep[p. 12]{Conway01}. Under this definition, the mean proportional growth rate is more negative if $x_i(t)$ is larger, or if $\Delta t$ is smaller, as required.

The size and shape components of the rate of change for a community can then be calculated as described in \ref{ap:ce}. The size component becomes a surreal number with an infinite and a real part:
\begin{equation}
\bar{r}(t, t + \Delta t) = \frac{(k_3-k_2)\psi}{n \Delta t} + \frac{a - b}{n \Delta t},
\label{eq:sizeinf}
\end{equation}
where $k_3$ is the number of colonizations, $k_2$ is the number of extinctions, $\psi = \omega^{1/\omega}$ is the negative of the natural logarithm of $1/\omega$ (and is larger than all positive real numbers), $a$ is the sum of log abundances for all species with nonzero abundance at time $t+\Delta t$, and $b$ is the sum of log abundances for all species with nonzero abundance at time $t$. Equation \ref{eq:sizeinf} reduces to Equation \ref{eq:rbar} if there are no colonizations or extinctions. The coefficient of $\psi$ in Equation \ref{eq:sizeinf} is the among-species mean of a variable taking the values $-1/\Delta t$ for extinctions, $1/\Delta t$ for colonizations, and zero otherwise (\ref{ap:pa}). If the difference between the numbers of colonizations and extinctions is not zero, the magnitude of this term will always be greater than the magnitude of the second term (the real part, proportional to $a-b$). To understand how Equation \ref{eq:sizeinf} works, consider a case where a community loses one species but gains another, whose final abundance is the same as the initial abundance of the lost species. No other species change abundance. Intuitively, there has been no change in size. If the species gained has lower final abundance than the species lost, there has been a reduction in size, but not as large as if a species was lost and no new species, however rare, colonized.

If there are extinctions and colonizations, the shape component also has infinite and real parts:
\begin{equation}
s_r \simeq \frac{ \alpha^{1/2} \psi}{{(n-1)^{1/2}}\Delta t} + \frac{\beta}{{2 [\alpha(n-1)]^{1/2}}\Delta t},
\label{eq:sdinf}
\end{equation}
where 
\begin{equation}
\begin{aligned}
\alpha &= k_2 + k_3 - \frac{1}{n}(k_3-k_2)^2,\\
\beta &= 2\left( \sum_{i \in \mathcal S_2} \log x_i(t) + \sum_{i \in \mathcal S_3} \log x_i(t+ \Delta t) - \frac{1}{n}(a-b)(k_3-k_2) \right),
\end{aligned}
\label{eq:beta}
\end{equation}
and $\mathcal S_2$ and $\mathcal S_3$ are the sets of species going extinct and colonizing, respectively. Equation \ref{eq:sdinf} is approximate in the sense that it ignores terms whose magnitude is less than any positive real number. The coefficient of $\psi$ in Equation \ref{eq:sdinf} is the among-species standard deviation of a variable taking the values $-1/\Delta t$ for extinctions, $1/\Delta t$ for colonizations, and zero otherwise (\ref{ap:pa}). Thus the infinite part will be large if there is a lot of change in the set of species present ($k_2+k_3$ is large), but there is little change in the number of species present ($k_3 - k_2$ is small). The second term in Equation \ref{eq:sdinf} (the real part) will be large if the species going extinct had large abundances at time $t$ and/or the species colonizing had large abundances at time $t+\Delta t$. 

The scaled activity level and the proportion of change that is size change in the presence of colonizations and extinctions can be calculated as described in \ref{ap:ce}. \R{A32}Because standard probability distributions are not defined over the whole of the surreal number field, and there is so far no generally satisfactory definition of an integral for surreal functions \citep{Fornasiero03, Rubinstein-Salzedo14}, parametric statistical analyses of these quantities would be difficult. However, surreal numbers can be ranked unambiguously (\ref{ap:ce}), and so statistics based on ranks of the scaled activity level and of the proportion of change that is size change are straightforward.

If only presence/absence data are available, only the infinite terms in Equations \ref{eq:sizeinf} and \ref{eq:sdinf} can be calculated. The coefficients of these terms are natural measures of size and shape change for presence/absence data, \R{A66}and it is not necessary to make any explicit use of the surreal number system to calculate them. \ref{ap:pa} summarizes the calculation of scaled activity and proportion of change that is size change from presence-absence data. The measures of rate of succession from presence/absence data defined by \citet{Anderson07} are based on rescalings of $k_2 + k_3$, which is proportional to the squared activity level from presence/absence data (\ref{ap:pa}).

\section{Example: the higher plant community on Surtsey}

I calculated the rate of shape change $s_r$ (Equation \ref{eq:sdinf}) and the among-species mean proportional growth rate $\bar{r}$ (Equation \ref{eq:sizeinf}) for the higher plant community on the volcanic island of Surtsey, using data and background information from \citet{Fridriksson89}. The island was formed in 1963, and its higher plant community has been surveyed annually since the first plant was found in 1965, initially by complete census, and later by quadrat and transect samples \citep{Fridriksson87}. Relatively few plant species are found on the island, because of the scarcity of water, low nutrient levels, salt spray, sand abrasion, and wave action \citep{Fridriksson89}. However, soil formation has been fairly rapid, due to organic matter being washed ashore \citep{Fridriksson87}. I used the data for 1965 to 1981 in Table 1 of \citet{Fridriksson89}. Data from later years were excluded because there was no count for the most abundant species after 1981. By 1981, 22 taxa had been observed, including one that was identified only to genus and a category of unidentified plants. There were many zero values.

It is plausible that abundances were exactly known, and therefore that zeros represent true absences (at least during the early part of the study, when complete censuses were made). The counts themselves can only take natural number values, so cannot be differentiable functions of time. Nevertheless, it is reasonable to treat expected abundances conditional on the data (which coincide with the counts at sampling times) as differentiable functions of time, and therefore to calculate the rates of change in size and shape using Equations \ref{eq:sizeinf} and \ref{eq:sdinf} respectively.

Figure \ref{fig:Surtseysurreal}A shows the logs of observed counts $c_i + 1$ against year, with the same interpretation as Figure \ref{fig:hoverflies}A when abundances were nonzero \R{A42}(again, the $\log(c_i + 1)$ transformation is used only for plotting, not for analysis). However, colonization and extinction were major features of these data. Species arrived at different times, with dispersal by birds probably the most common route \citep{Fridriksson87}. Of these arrivals, most did not persist or remained at low abundance, some increased in abundance and then decreased again, and one (the sandwort \emph{Honckenya peploides}, Figure \ref{fig:Surtseysurreal}A, black line) became numerically dominant.

\begin{figure}
\includegraphics[height=18cm]{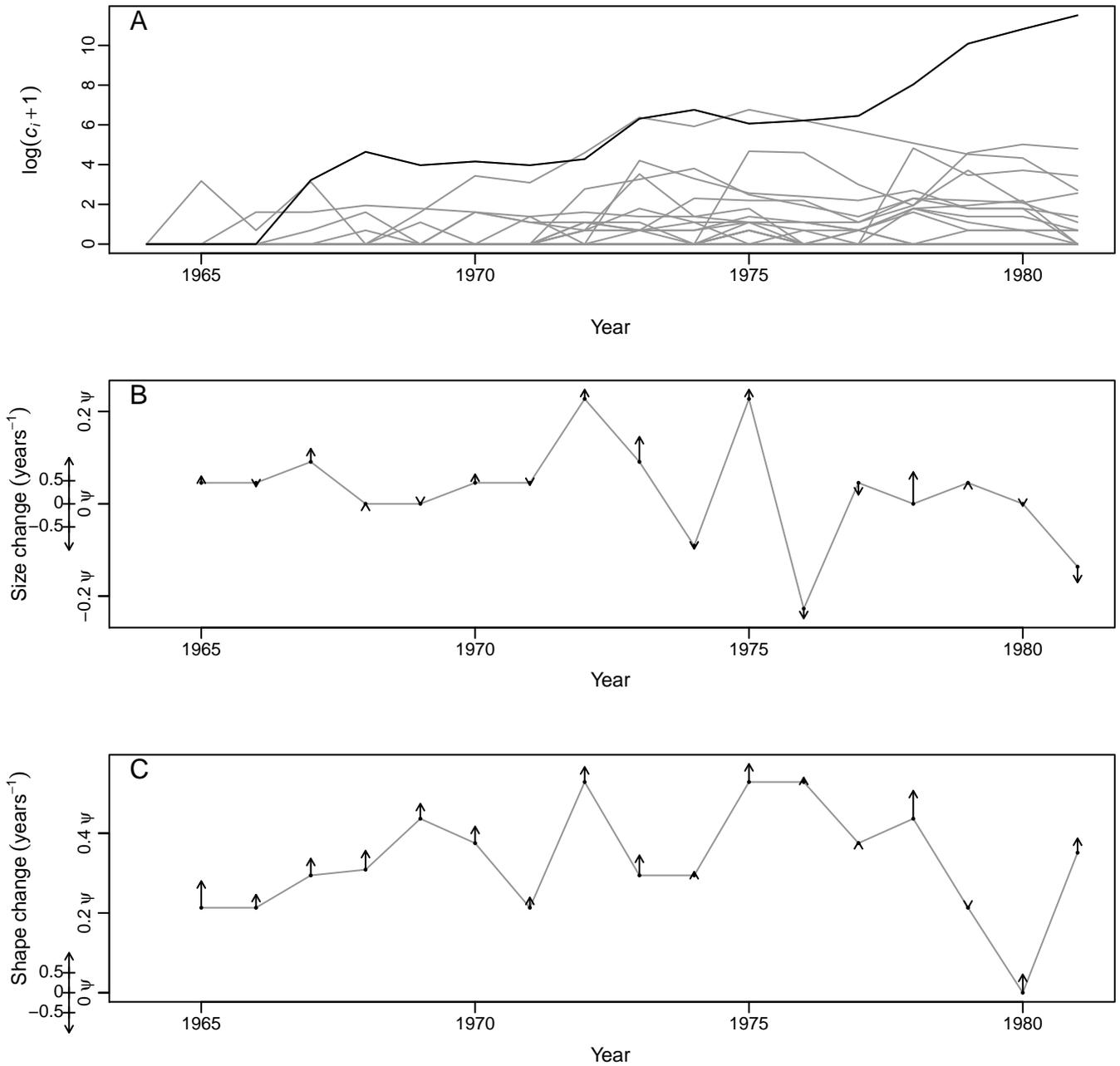}
\caption{Higher plant community on Surtsey, 1965 to 1981. A: log of observed count $c_i+1$ against time for each species. Black line is the dominant species \emph{Honckenya peploides}, grey lines are other species. B: rate of change in size (among-species mean of mean proportional growth rates $\bar{r}$, years$^{-1}$) against time. Dots represent the infinite part (measured in units of $\psi$, main scale) and lengths of arrows represent the real part (secondary scale, magnified by $0.1\psi$). C: rate of change in shape (among-species standard deviation of mean proportional growth rates, $s_r$, years$^{-1}$) against time. Infinite and real parts represented as in B.}
\label{fig:Surtseysurreal}
\end{figure}

Equation \ref{eq:sizeinf} gives the among-species mean $\bar{r}$ of mean proportional growth rates, including colonizations and extinctions. In Figure \ref{fig:Surtseysurreal}B, the dots connected by lines represent the infinite part of $\bar{r}$, while the lengths of the arrows represent the real part (on a separate scale, magnified by a factor of $0.1\psi$). The most striking feature of the temporal pattern in $\bar{r}$ is that it does not \R{A63}increase systematically over time, despite changes in the environment that might be expected to favour plant growth. For example, over the last two time periods, the abundance of the dominant species \emph{Honckenya peploides}, and the total number of individuals, were increasing rapidly. Biomass production and soil fertility, driven by inputs of fish waste and guano from seabirds, were also increasing \citep{Fridriksson89}. Furthermore, as the abundance of \emph{H. peploides} increased, dense stands of this plant modified the environment in ways that increased the germination and survival rates of two other species \citep{Fridriksson87}, an example of the type of positive interaction thought to be important in facilitating succession \citep{Bertness94}. Nevertheless, both the infinite and real parts of $\bar{r}$ were lower in the last two time periods than in much of the first half of the series. In other words, for a typical species, the environment was less favourable at the end than the start of the time series.

The rate of change in shape, $s_r$ (Equation \ref{eq:sdinf}) had relatively high values in the middle of the study period (Figure \ref{fig:Surtseysurreal}C). Inspection of Figure \ref{fig:Surtseysurreal}A shows that in the middle part of the series, there were relatively large numbers of colonizations and extinctions (values of $\log (c_i + 1)$ changing to and from zero). Simultaneous colonizations and extinctions tend to make the infinite part of $s_r$ large (Figure \ref{fig:Surtseysurreal}C, dots connected by lines). There was little obvious temporal pattern in the real part of $s_r$ (Figure \ref{fig:Surtseysurreal}C, arrows, on a separate scale magnified by a factor of $0.1\psi$). Figure \ref{fig:Surtseysurreal}A shows that at most time points, there was substantial among-species variability in changes in abundance. At least over this time scale, the idea that rates of shape change are \R{A54}decreasing over time is not supported. In retrospect this is unsurprising, because as Figures \ref{fig:succession}C and D show, it is easy to construct cases in which rates of change are not decreasing over time in the short term.\R{A60}

The amount of change in the Surtsey plant community was almost always not commensurate with that in the Leicester hoverfly assemblage, in the sense that colonization and extinction represent infinitely larger changes than any change in abundance not involving such events. Nevertheless, the proportion of change that is size change ($\cos^2 \theta$) can be directly compared (Figure \ref{fig:hovervsSurtseybox}). This proportion was strikingly higher for the Leicester hoverfly assemblage (median 0.22, lower quartile 0.07, upper quartile 0.35) than for the Surtsey higher plant community (median 0.05, lower quartile 0.02, upper quartile 0.09), although a hypothesis test is not appropriate because no hypothesis was proposed before looking at the data. In general, one might expect phylogenetically restricted assemblages of organisms such as the Leicester hoverflies to show larger amounts of size change than more diverse assemblages of organisms such as the Surtsey higher plants. This is because if there is strong niche conservatism \citep{Wiens10}, closely-related species will tend to have similar proportional growth rates under most sets of environmental conditions, and when proportional growth rates are similar, $\cos^2 \theta$ will tend to be large (Figure \ref{fig:projection}). Conversely, in a diverse assemblage, it is unlikely that all species will have similar proportional growth rates under most sets of environmental conditions, and $\cos^2 \theta$ will generally be small. Thus, in diverse assemblages, indices of size change such as the Living Planet Index may be relatively uninformative about community dynamics. \R{A36}I will not attempt to test this hypothesis here, because such a test will require data from a wide range of assemblages of differing taxonomic diversity, \R{A58}but it is closely connected to the idea that compensatory dynamics (predominantly positive interspecific covariances in abundance) will occur in assemblages where many species are functionally equivalent \citep{Houlahan07}.

\begin{figure}
\includegraphics[height=18cm]{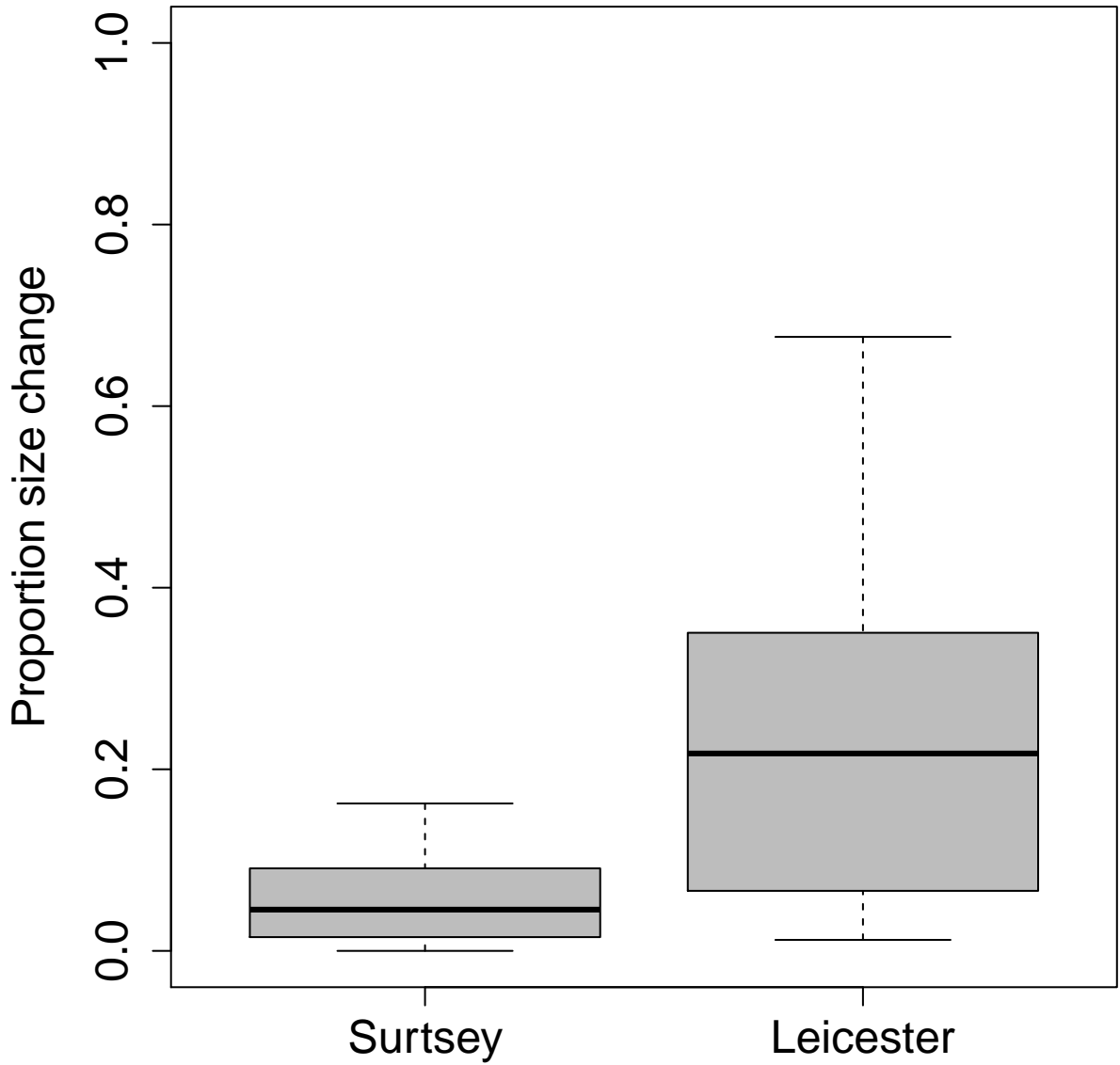}
\caption{Boxplots of proportion of change that is size change ($\cos^2 \theta$) for higher plants on Surtsey (17 annual time intervals: infinitesimal parts not shown) and hoverflies in a Leicester, UK, garden (29 annual time intervals, bootstrap mean values plotted). Thick line: median. Box extends from the lower to the upper hinges, which are approximately the lower and upper quartiles. Whiskers extend from the ends of the box to the furthest observations no more than 1.5 times the width of the box away.}
\label{fig:hovervsSurtseybox}
\end{figure}

In summary, the Surtsey higher plant data illustrate how patterns in size change may not accord with simple perceptions of the favourability of an environment, and do not support the idea that rates of shape change will generally be decreasing over time. Change in shape was much more important in the Surtsey higher plant assemblage than in the Leicester hoverfly assemblage.

\section{Discussion}

Here, I have derived a measure of the rate of shape change (in the sense of Definition \ref{def:ratesuc}) that is consistent with the basic principles of population dynamics and takes an organism-centred view. \R{A75}The most important result is that if two abundance trajectories have the same proportional growth rates at corresponding time points, they are equivalent from an organism's point of view, and should therefore have the same rate of shape change. It may seem surprising that equal proportional changes make equal contributions to the rate of shape change (and to indices of biodiversity change such as the Living Planet Index) whether they occur in rare or common species. Thus, shape change could be rapid in a community whose dominant species do not change in abundance. Population biologists have accepted the idea that the proportional scale is appropriate for analyses of single-species dynamics \citep{Gaston_1893}. A reluctance to apply the same logic at the community level \citep[e.g.][]{Lewis78, Field82, Legendre01} might imply that some variable other than the composition of the community is in fact of primary interest. Such variables might include total biomass or its distribution among species, nutrients, phylogenetic diversity, or structural complexity. These are all important aspects of a broader concept of succession. Nevertheless, one should not object to a measure of the rate of shape change on the grounds that it does not describe the rate of change of some other variable. The properties studied by community ecologists (such as abundances) and the properties studied by ecosystem ecologists (such as biomass, nutrients, and structural complexity) may behave in quite different ways.

\R{A50}Other interesting measures of community dynamics could be obtained under different invariance principles. For example, one might argue that for a single population, all points on the same logistic growth curve are equivalent in the long term. Under this argument, one would want a measure (such as the carrying capacity) that indexes long-term behaviour rather than short-term dynamics. There are two key differences from the approach taken in this manuscript. First, situations that will be equivalent in the long term are not equivalent from an organism's point of view. Instead, they might represent equivalence in properties such as long-term conservation value. Second, such an approach requires a parameterized model of community dynamics.  There are fairly simple and general stochastic models of community dynamics \R{A61}\citep[e.g.][]{Ives03,Mutshinda09,Hampton13}, but they will not apply to all situations. The model would (I think) have to be parameterized, because it would be difficult to find a useful measure that was invariant over the space spanned by all members of a family of models, rather than the space spanned by a particular member of that family. There is a similar issue with measures of physiological time. It is easy to find a physiological time scale on which some process such as organism growth has a constant proportional rate, but impossible to find a universal time scale on which all relevant processes will have this property \R{A51}\citep{vanStraalen83}. In contrast, the approach used here does not rely on any parameterized model of community dynamics (other than in the trivial sense that because we do not usually know abundances exactly, we have to use a model to estimate proportional growth rates).

\R{A57}There may be a deeper connection between measures of diversity and measures of the rate of shape change. Diversity indices are scalar summaries of a relative abundance vector, while measures of the rate of shape change are scalar summaries of changes in relative abundance vectors. Most diversity indices can be expressed in terms of Hill numbers, a family of functions indexed by the weight they give to evenness \citep{Hill73,Jost06}. There are generalizations that account for functional or taxonomic similarity \citep{Leinster12,Chiu14}. The obvious place to look for a connection between the concepts of diversity and rate of shape change is beta diversities, which measure differences in diversity between communities. Beta diversities based on Hill numbers unify a wide range of measures of community similarity \citep{Jost07}, including the Jaccard and Morisita-Horn indices used by \citet{Dornelas14} to measure temporal change in communities. They are not scaling invariant (M. Spencer, unpublished results, unless no weight is given to evenness, in which they only measure colonization and extinction), but the possibility of a less obvious connection remains.

Traditionally, multivariate community ecology has focussed on variation among sites \citep[e.g.][]{Gauch82}. Even when samples form a time series, the usual approach is to analyze them in ``abundance space'', where each observation represents the abundance vector at a single time \citep[e.g.][]{Warwick02, Lear08}. This is implicitly a static view of communities, in which abundance is directly determined by environmental conditions. It is in conflict with the dominant view in population dynamics that environmental conditions act on the proportional growth rate of a species, rather than directly on its abundance \citep{Birch53, Hutchinson57, Maguire73}. \R{A31}In contrast, the decomposition of patterns of change into size and shape components gives an ecologically meaningful low-dimensional representation of dynamics, which may offer a useful alternative to traditional ordination. Furthermore, it can be extended to deal with colonizations and extinctions, or to data on presence/absence alone. \R{A65}Related work includes the regression method used by \citet{Ives95} to study the responses of communities to long-term directional perturbations, although this focusses on responses of mean abundance to environmental change, rather than on decomposing empirical patterns of proportional growth rates.

What could measures of size and shape change be used for? One obvious idea is to quantify global patterns in community change. There is evidence of global variation in size change in vertebrate assemblages since 1970 \citep{Collen09,Collen13}. For example, among terrestrial species, tropical populations appear to be declining while temperate populations are increasing. For marine species, populations in the Southern and Indian Oceans appear to be declining, while populations in the Pacific and Atlantic are relatively stable. \R{A17}Thus, there is clearly shape change as well as size change at the global scale, and some of this change can be explained by simple geographical differences. Similarly, among insects there was an overall 45$\%$ decline in abundance between 1970 and 2010, but declines were less severe in Lepidoptera than in other insects \citep{Dirzo14}. The geometric approach taken here shows that we can think of the Living Planet Index and related indices as being based on a one-dimensional linear model \citep[e.g.][chapter 5]{Saville91}. \R{A43}Rather than fitting separate one-dimensional models for subgroups, as is currently done \citep{Collen13}, it would be productive to fit a single linear model with an overall size change pattern, differences between groups (e.g. tropical and temperate populations) represented as orthogonal contrasts \citep[e.g.][section 7.3]{Saville91}, and a residual shape change component. This will allow quantification of the amount of change that can be explained by simple differences such as tropical vs temperate, and the amount of unexplained but ecologically meaningful shape change. The fact that abundances do not need to be measured in the same units for different species moves the task of studying such global patterns from impossible to merely very difficult. Furthermore, it may be possible to calculate the analogous measures for presence/absence data in cases where abundances cannot be reliably estimated. For example, it would be possible to calculate size and shape change indices from historical records and fossil communities, and hence study global patterns of change over very long time scales.

In the two communities I analyzed, there was a striking difference in the proportion of change that was size change. It therefore seems likely that size-based indices such as the Living Planet Index will be much more informative about some communities than others. I hypothesized that phylogenetically-restricted assemblages might show mainly size change, while diverse assemblages might show mainly shape change. For example, hydrothermal vent communities occur are characterized by evolutionary radiations of a fairly small group of taxa \citep{VGS+02}. It seems plausible that such communities might show a relatively high proportion of size change. In contrast, one might expect tropical forests and coral reefs, with high phylogenetic diversity, to show a high proportion of shape change.

\R{A71}In conclusion, shape change is distinct from but complementary to the property measured by indices such as the Living Planet Index. The rate of shape change should be measured in a way that is consistent with population dynamics. The measures I have derived satisfy this requirement. The key to understanding their properties is the geometry of growth space, whose axes are the proportional growth rates of each species involved.

\section*{Acknowledgements}

I am grateful to six anonymous reviewers for constructive criticism of earlier versions of this manuscript. Much of this work was done while MS was a visitor at the Department of Mathematics and Statistics, University of Otago, and a Sabbatical Fellow at the National Institute for Mathematical and Biological Synthesis, an Institute sponsored by the National Science Foundation, the U.S. Department of Homeland Security, and the U.S. Department of Agriculture through NSF Award \#EF-0832858, with additional support from The University of Tennessee, Knoxville. Further funding was provided by the School of Environmental Sciences, University of Liverpool, and by NERC grant NE/K00297X/1 awarded to MS. Discussions with Kath Allen, David Bryant, Damian Clancy, Jennifer Cooper, Lev Ginzburg, Lou Gross, David Gurarie, Axel Rossberg, and Kamila \.Zychaluk contributed to the development of the ideas described here.

\clearpage
\begin{appendices}
\gdef\thesection{Appendix \Alph{section}}
\numberwithin{equation}{section}
\renewcommand{\theequation}{\Alph{section}.\arabic{equation}}

\section{Invariance under exponential growth and scaling invariance}
\label{ap:invariance}

\begin{theorem} A measure of the rate of shape change satisfying Property \ref{property:perturbinvar} is invariant under exponential growth, in the sense that it will have the same value for any pairs of vectors $\mathbf x$, $\mathbf y$ generated by a given initial abundance vector $\mathbf x_0$, a constant vector of species-specific proportional growth rates $\mathbf r$ and a constant time interval $\Delta t$. In other words,
\begin{equation*}
g(\mathbf x, e^{\mathbf R \Delta t} \mathbf x) = g(e^{\mathbf R t} \mathbf x, e^{\mathbf R(t+ \Delta t)} \mathbf x).
\end{equation*}
\end{theorem}

\begin{proof}
Property \ref{property:perturbinvar} is that
\begin{equation*}
g(\mathbf x, \mathbf y) = g(\mathbf M \mathbf x, \mathbf M \mathbf y)
\end{equation*}
for all $n \times n$ matrices $\mathbf M$ with positive diagonal elements and zero off-diagonal elements. Choose $\mathbf x= \mathbf x_0$, $\mathbf y = e^{\mathbf R \Delta t} \mathbf x$ where $\mathbf R = \text{diag}(\mathbf r)$, and $\mathbf M = e^{\mathbf R t}$ for any time $t$. Then by Property \ref{property:perturbinvar},
\begin{equation*}
g(\mathbf x_0, e^{\mathbf R \Delta t} \mathbf x_0) = g(e^{\mathbf R t} \mathbf x_0, e^{\mathbf R(t+ \Delta t)} \mathbf x_0)
\end{equation*}
for any $t$.
\end{proof}

\begin{theorem}
Invariance under exponential growth is necessary for Property \ref{property:perturbinvar}. 
\end{theorem}

\begin{proof}
Suppose that there exist a set of proportional growth rates $\mathbf r$ and a time interval $\Delta t$ for which the measure is not invariant. In other words, for some initial vector $\mathbf x$, $\mathbf y = e^{\mathbf R \Delta t} \mathbf x$ and $\mathbf z = e^{\mathbf R \Delta t} \mathbf y$, but $g(\mathbf x, \mathbf y) \neq g(\mathbf y, \mathbf z)$. Then for $\mathbf M = e^{\mathbf R \Delta t}$, Property \ref{property:perturbinvar} does not hold.
\end{proof}

\begin{theorem}
Invariance under exponential growth is not sufficient for Property \ref{property:perturbinvar}. For example, a measure which is invariant under exponential growth from any given initial conditions, but with a different value for each initial condition, does not satisfy Property \ref{property:perturbinvar}. 
\end{theorem}

\begin{proof}
Suppose that for any initial condition $\mathbf x$, $\mathbf y = e^{\mathbf R \Delta t} \mathbf x$. Choose $\mathbf x_1$ so that there is no $t$ for which $\mathbf x_1 = e^{\mathbf R t} \mathbf x$, and $\mathbf y_1 = e^{\mathbf R \Delta t} \mathbf x_1$. Suppose that $g(\mathbf x, \mathbf y) \neq g(\mathbf x_1, \mathbf y_1)$, then Property \ref{property:perturbinvar} does not hold for $\mathbf M$ such that $\mathbf x_1 = \mathbf M \mathbf x$.
\end{proof}

\clearpage

\section{Existing measures of rate of shape change}
\label{ap:existing}

Here, I review some widely-used measures of the rate of shape change. The original authors did not use the term shape change. Instead, they used terms such as ``succession rate'', ``rate of change'' and ``compositional dissimilarity'' to describe what they were measuring. Nevertheless, they all fit Definition \ref{def:ratesuc}.

The most obvious choice for a measure of the mean rate of shape change \citep{Foster00} is $1/\Delta t$ times the Euclidean distance between the relative abundance vectors at times $t$ and $t + \Delta t$,
\begin{equation}
\frac{1}{\Delta t} \left[ \sum_{i=1}^n (p_i(t + \Delta t) - p_i(t))^2 \right ] ^{1/2}.
\label{eq:euclidean}
\end{equation}
As a physical analogy, this is the way one would calculate the mean speed of an object in space, given its position at two points in time. However, although Euclidean distance is the natural choice for points in the real space $\mathbb R^n$, relative abundances lie in the unit simplex $\mathbb S^{n-1}$, for which Euclidean distance is not a natural choice. For example, the Euclidean distance between $[0.5, 0.5]^\prime$ and $[0.4, 0.6]^\prime$ is the same as that between $[0.1, 0.9]^\prime$ and $[0.2, 0.8]^\prime$, but the latter case represents a much larger proportional change. Since exponential growth corresponds to constant proportional change, Euclidean distance is not invariant under exponential growth (Figure \ref{fig:succession}A, purple line).

In an insightful paper, \citet{Jassby74} pointed out that an absolute abundance vector can be mapped to a point $\mathbf e$ on the unit $s$-sphere by the transformation $e_i(t) = x_i(t)/[\sum_{i=1}^n (x_i(t))^2]^{1/2}$. As the community changes, it traces out a path on the sphere of length
\begin{equation*}
S(t) = \int_0^t \left | \frac{\mathrm d \mathbf e}{ \mathrm d u}\right | \, \mathrm d u.
\end{equation*}
They therefore proposed
\begin{equation}
\sigma = \mathrm d S / \mathrm d t = \left[ \sum_{i=1}^n (\mathrm d e_i/ \mathrm d t)^2 \right]^{1/2}
\label{eq:sigma}
\end{equation}
as a measure of the instantaneous rate of shape change. They estimated $\sigma$ using a finite-difference approximation, replacing $\mathrm d e_i / \mathrm dt$ by $\Delta e_i / \Delta t$, for small $\Delta t$. However, their measure is not invariant under exponential growth (Figure \ref{fig:succession}A, red line).

\citet{Lewis78} argued (wrongly, in my view) that $\sigma$ does not weight each species equally, and proposed
\begin{equation}
\sigma_s = \frac{1}{\Delta t}\sum_{i=1}^n | p_i(t + \Delta t) - p_i(t) |
\label{eq:Lewis}
\end{equation}
(over sufficiently small time intervals $\Delta t$) as a more equitable measure of the instantaneous rate of shape change. The community is represented as a point on the surface of a polygon embedded in $\mathbb R^n$, and distance is measured as the sum of absolute distances along each dimension. This is a rescaling of the taxicab metric, which is discussed extensively by \citet{Miller02}. While it has a simple geometric interpretation, it does not measure proportional change, and is not invariant under exponential growth (Figure \ref{fig:succession}A, blue line).

The Bray-Curtis distance is widely used as a measure of dissimilarity between communities, \R{A73}and changes in Bray-Curtis distances among samples in a sequence are sometimes interpreted as changes in the rate of shape change \citep[e.g.][]{Nicholson79}. It is particularly popular with marine ecologists studying changes in communities \citep{Field82}. Most of this marine ecological work has been based on the ranks of distances rather than the distances themselves, but clearly, a rank-based approach relies to some extent on the properties of the underlying distance. Denote by $\kappa_i(t) = h(x_i(t))$ the score for species $i$ at time $t$, where $h$ is some suitable transformation. Then a Bray-Curtis-based measure of the rate of change in scores per unit time is
\begin{equation}
\frac{\sum_{i=1}^n | \kappa_i(t + \Delta t) - \kappa_i(t) |}{\Delta t \sum_{i=1}^n [ \kappa_i(t + \Delta t) + \kappa_i(t) ] }
\label{eq:bc}
\end{equation}
\citep{Field82}. If $\kappa_i(t) = p_i(t)$, then this is just $1/2$ times Lewis's $\sigma_s$ (Equation \ref{eq:Lewis}). However, \citet{Field82} recommend a transformation that reduces the weighting given to very abundant species, such as the fourth root. For relative abundances, this is $\kappa_i(t) = [p_i(t)]^{1/4}$. Whether the transformation is applied to relative or to absolute abundances, the result is not invariant under exponential growth (Figure \ref{fig:succession}A, green line, using relative abundances).

Plant ecologists often use methods based on correspondence analysis to search for relationships between environmental variables and community structure \citep{terBraak85}. Correspondence analysis preserves the chi-square distance between communities:
\begin{equation*}
d_{\chi^2}(t, t+ \Delta t) = (x_{..})^{1/2} \left[\sum_{i=1}^n \frac{1}{x_{.i}} \left(p_i(t + \Delta t) - p_i(t)  \right)^2 \right] ^{1/2},
\
\end{equation*}
where $x_{..}$ is the sum of abundances over all species and times, and $x_{.i}$ is the sum of abundances of the $i$th species over all times  \citep{Legendre01}. The desire to interpret distances on a correspondence analysis biplot as measures of the amount of change leads to
\begin{equation}
\frac{d_{\chi^2}(t,t+\Delta t)}{\Delta t}
\label{eq:dchi2}
\end{equation}
 as a possible measure of the rate of shape change. Equation \ref{eq:dchi2} is not invariant under exponential growth (Figure \ref{fig:succession}A, orange line). Detrended correspondence analysis \citep{Hill80} is one of the most popular methods based on correspondence analysis, and distances on a detrended correspondence analysis plot are sometimes viewed as amounts of change \citep[e.g.][p. 253]{Jacobson86, Bush04, Walker03}. I do not discuss detrended correspondence analysis in detail, but the additional ad-hoc operations it involves do not alter the basic conclusion that chi-square distances are not invariant under exponential growth.

The patterns of change over time in Equations \ref{eq:sigma} to \ref{eq:dchi2} in a community of exponentially-growing species depend on the proportional growth rates and initial abundances of the species, so that they may in other cases show patterns quite different from those in Figure \ref{fig:succession}A. I therefore do not attempt to make any statements about their relative usefulness, other than that changes in their values cannot be used to identify ecologically meaningful changes in the rate of shape change. 

\clearpage
\section{Deriving a measure of rate of shape change}
\label{ap:derive}

Here, I derive a measure of the rate of shape change which has Properties \ref{property:ratefunc} to \ref{property:richness}, for the case where there are no colonizations or extinctions.

I start with an analogous but simpler problem. Consider a single species with abundance $x(t)$ at time $t$, where $x(t)$ is a differentiable, positive function of time. Denote by $\tilde{r}(t,t+\Delta t)$ the mean proportional growth rate over the time interval $(t,t + \Delta t]$, which is
\begin{equation*}
\tilde{r}(t,t+\Delta t) = \frac{1}{\Delta t} \int_t^{t+\Delta t} \frac{1}{x} \frac{\mathrm d x}{ \mathrm d u} \, \mathrm du = \frac{\log x(t + \Delta t) - \log x(t)}{\Delta t}, \quad x>0,
\end{equation*}
(throughout, $\log$ denotes the natural logarithm). This mean proportional growth rate is obviously a constant for all patterns of change resulting in the same amount of proportional growth in a given time.

I will now derive a similar measure of the mean rate of proportional change in relative abundances (thus satisfying Properties \ref{property:ratefunc} and \ref{property:timefunc}). All the information in relative abundance data is contained in the set of ratios $v_{ij}=x_i/x_j$, $i=1,\ldots,n-1$, $j=i+1,\ldots,n$, satisfying Property \ref{property:scaleinvar}). In other words, a basis for the space of relative abundances can be constructed from functions of these ratios \citep{Egozcue03}. We therefore need only consider the vector of ratios $\mathbf v(t) = [v_{12}(t), v_{13}(t), \ldots, v_{1n}(t), v_{23}(t),\ldots v_{n-1,n}(t)]^{\prime}$. \R{A6}Although $\mathbf v(t)$ is a generating set for the space of relative abundances, it is not a minimal generating set (basis). Nevertheless, treating all its elements symmetrically makes it easy to see the biological interpretation of the resulting measure, and results in the same measure as would be obtained from a basis.

There are $c=n(n-1)/2$ elements in $\mathbf v(t)$, but their order is not important in what follows. Denote by $w_k$ the mean proportional rate of change in the $k$th element of $\mathbf v$ (which I index by $v_{ij}$) over the time interval $(t,t + \Delta t]$, and denote by $\tilde{r}_i(t,t+\Delta t)$ the mean proportional growth rate of the $i$th species over the time interval $(t,t+\Delta t]$, defined as in Equation \ref{eq:meanrate}. Then, assuming $v_{ij}$ is a differentiable function of time,
\begin{equation*}
\begin{aligned}
w_k &= \frac{1}{\Delta t} \int_t^{t + \Delta t} \frac{1}{v_{ij}}\frac{\mathrm d v_{ij}}{\mathrm d u }\, \mathrm d u, \quad v_{ij} > 0,\\
&= \frac{\log v_{ij}(t + \Delta t)-\log v_{ij}(t)}{\Delta t}\\
&= \frac{1}{\Delta t}\left[\log \frac{x_i(t+\Delta t)}{x_j(t+\Delta t)} - \log \frac{x_i(t)}{x_j(t)} \right]\\
&= \frac{1}{\Delta t} [ (\log x_i(t+\Delta t) - \log x_i(t))- (\log x_j(t+\Delta t) - \log x_j(t)) ] \\
&=\tilde{r}_i(t,t+\Delta t) - \tilde{r}_j(t,t+\Delta t).
\end{aligned}
\end{equation*}
This expression will be zero for any $i$ and $j$ if $\mathbf x(t + \Delta t) = \mathbf x (t)$ (Property \ref{property:zero}), or if $\mathbf x(t + \Delta t) = a \mathbf x (t)$, where $a > 0$ (Property \ref{property:neutral}). Because $w_k$ is the difference between mean proportional growth rates, any pattern of growth that leads to the same proportional change in species $i$ and $j$ over a given time interval will result in the same value of $w_k$, and Property \ref{property:perturbinvar} is satisfied.

The obvious scalar measure of the mean rate of proportional change in relative abundances is then the Euclidean norm of the vector $\mathbf w = [w_1,w_2,\ldots,w_c]^{\prime}$, which is
\begin{equation}
\begin{aligned}
\| \mathbf w \| &= \left[ \sum_{k=1}^c w_k^2 \right] ^{1/2}\\
&= \frac{1}{\Delta t} \left[ \sum_{i=1}^{n-1} \sum_{j=i+1}^n \left[ \log v_{ij}(t + \Delta t) - \log v_{ij}(t) \right]^2 \right] ^{1/2}\\
&= \left[ \sum_{i=1}^{n-1} \sum_{j=i+1}^n \left[ \tilde{r}_i(t,t+\Delta t) - \tilde{r}_j(t,t+\Delta t) \right]^2 \right]^{1/2}.
\end{aligned}
\label{eq:dA}
\end{equation}
Squaring each element in $\mathbf w$ ensures that the measure will be non-negative (Property \ref{property:nonneg}).

The second line of Equation \ref{eq:dA} expresses $\| \mathbf w \|$ in a way that shows it is proportional to the Aitchison distance \citep{Aitchison92, Egozcue03} between the relative abundance vectors at times $t$ and $t + \Delta t$\R{A22}. The Aitchison distance is in some sense the simplest measure of difference between two compositions that satisfies Properties \ref{property:scaleinvar}, \ref{property:zero} and \ref{property:perturbinvar}, as well as the additional requirements of symmetry and invariance under arbitrary reorderings of the species \citep{Aitchison92}. It also satisfies the triangle inequality. The third line of Equation \ref{eq:dA} expresses $\| \mathbf w \|$ in a way that shows it is a function of the mean proportional growth rates of all the species.

To obtain a measure that is independent of the number of species (Property \ref{property:richness}), I will use a geometric interpretation of Equation \ref{eq:dA}. For a given set of species, every vector of proportional growth rates $\mathbf r$ (dropping the time indexing for simplicity) can be represented as a point in a space $\mathbb R^n$ which I refer to as the growth space of the community (Figure \ref{fig:projection} is a two-dimensional example). The growth rate vector $\mathbf r$ can be decomposed into two orthogonal components. The first component is the projection $\mathbf u$ of $\mathbf r$ onto the line of equal proportional growth rates, which represents change in abundance without shape change. I call this size change. \R{A16}It is straightforward \citep[e.g.][p. 69]{Saville91} to show that $\mathbf u = \bar{r} \mathbf 1$, where\R{A21}
\begin{equation}
\bar{r} = \frac{1}{n} \sum_{i=1}^n \tilde{r}_i,
\label{eq:lpi}
\end{equation}
with dimensions time$^{-1}$. In other words, each element of $\mathbf u$ is the among-species sample mean of mean proportional growth rates. The second component $\mathbf q$ is orthogonal to $\mathbf u$, and the square of its Euclidean norm is \R{A19}
\begin{equation*}
\| \mathbf q \|^2 = \sum_{i=1}^n (\tilde{r}_i - \bar{r})^2,
\end{equation*}
\citep[e.g.][pp. 49-50]{Saville91}. It is straightforward to show that $\| \mathbf w \|^2 = n \| \mathbf q \|^2$:
\begin{equation*}
\begin{aligned}
\| \mathbf w \|^2 &= \sum_{i=1}^n \sum_{j=i+1}^n \left[ \tilde{r}_i - \tilde{r}_j \right]^2 \\
&= \frac{1}{2} \sum_{i=1}^{n-1} \sum_{j=1}^n \left[ \tilde{r}_i - \tilde{r}_j \right]^2 \\
&= \frac{1}{2} \sum_{i=1}^n \sum_{j=1}^n \left[ (\tilde{r}_i-\bar{r}) - (\tilde{r}_j - \bar{r}) \right]^2 \\
&= \frac{1}{2}\sum_{i=1}^n \sum_{j=1}^n \left[ (\tilde{r}_i-\bar{r})^2 - 2(\tilde{r}_i -\bar{r})(\tilde{r}_j - \bar{r}) + (\tilde{r}_j-\bar{r})^2 \right] \\
&= n \sum_{i=1}^n (\tilde{r}_i-\bar{r})^2 - \sum_{i=1}^n \left[ (\tilde{r}_i -\bar{r}) \sum_{j=1}^n (\tilde{r}_j - \bar{r} ) \right] \\
&= n \sum_{i=1}^n (\tilde{r}_i-\bar{r})^2 - \sum_{i=1}^n \left[ (\tilde{r}_i -\bar{r}) \left( \sum_{j=1}^n \tilde{r}_j - n \bar{r} \right) \right] \\
&= n \sum_{i=1}^n (\tilde{r}_i-\bar{r})^2 - \sum_{i=1}^n (\tilde{r}_i -\bar{r}) \left( n \bar{r} - n \bar{r} \right) \quad \text{from Equation \ref{eq:lpi}}\\
&= n \sum_{i=1}^n (\tilde{r}_i-\bar{r})^2 - \sum_{i=1}^n (\tilde{r}_i -\bar{r}) \times 0 \\
&= n \| \mathbf q \| ^2.
\end{aligned}
\end{equation*}
In other words, the norm of $\mathbf q$ is proportional to the proposed scalar measure of mean rate of proportional change in relative abundances (Equation \ref{eq:dA}), which can be thought of as measuring change in shape of the community.

An obvious measure of rate of shape change that is independent of the number of species is therefore\R{A20}
\begin{equation}
\begin{aligned}
s_r &= \left [\frac{1}{n-1} \sum_{i=1}^n (\tilde{r}_i - \bar{r})^2 \right] ^{1/2}\\
&=\frac{1}{{(n-1)^{1/2}}} \| \mathbf q \| \\
&= \frac{1}{(n(n-1))^{1/2}} \| \mathbf w \|,
\label{eq:srqw}
\end{aligned}
\end{equation}
which is just the sample standard deviation of the mean proportional growth rates, with dimensions time$^{-1}$. Since the sample standard deviation is an \R{A45}asymptotically unbiased estimator of the population standard deviation of the mean proportional growth rates, the expected value of $s_r$ will (in the limit of a large number of species) be constant among communities with different numbers of species but the same population standard deviation of mean proportional growth rates, or among random samples of different numbers of species from a single community. Thus Property \ref{property:richness} is satisfied asymptotically in principle. However, the above argument assumes that the mean proportional growth rates are known exactly. In reality, the true abundances (and hence true mean proportional growth rates) will be uncertain. This uncertainty may affect how the value of $s_r$ depends on the number of species, but the way in which it does so will depend on how abundances are estimated. Furthermore, $s_r$ is biased for finite samples. In principle, this bias can be corrected, but only by assuming a particular distribution, such as the normal, for the proportional growth rates \citep[e.g.][]{Gurland71}.

\R{A78}To obtain Equation \ref{eq:srp}, use Equation \ref{eq:dA} with $v_{ij} = x_i/x_j= p_i/p_j$ and Equation \ref{eq:srqw}.

\clearpage

\section{What kinds of functions are scaling invariant?}
\label{ap:maxinvar}\R{A25}
It is worth asking whether there are many kinds of scaling invariant functions other than $s_r$, and if so, what they look like. The theorem below is essentially a rephrasing of \citet[p. 373]{Aitchison92}. It leads to the result that other scaling invariant functions take particular forms (such as contrasts), and cannot involve any kind of weighting by abundance.

\begin{theorem}
\label{theorem:maxinvar}
All scaling invariant functions are functions of the vector $\mathbf r$ of mean proportional growth rates over some fixed time interval.
\end{theorem}
\begin{proof}

Denote by $\mathcal Z$ the set of pairs of abundance vectors $z = (\mathbf x, \mathbf y) \in \mathbb R^n_{>0} \times \mathbb R^n_{>0}$. A scaling can be represented by an $n \times n$ matrix $\mathbf M$ with positive diagonal elements, that maps $(\mathbf x, \mathbf y)$ to $(\mathbf M \mathbf x, \mathbf M \mathbf y)$ (property \ref{property:perturbinvar}). Denote by $\mathbf Mz$ the action of $\mathbf M$ on $z$.

The matrices $\mathbf M$ form a group $\mathcal M$ under matrix multiplication \citep[problem 4.5.8]{Pollatsek09}. The orbit $O_z$ of any $z \in \mathcal Z$ is the set of elements of $\mathcal Z$ that can be reached by the action of any $\mathbf M \in \mathcal M$ on $z$, and is the set of pairs of abundance vectors equivalent to the pair $z=(\mathbf x, \mathbf y)$ under scaling \citep[p. 21]{Eaton89}.

A function $g$ is invariant if $g(z) = g(\mathbf Mz)$ for all $\mathbf M \in \mathcal M$, and maximal invariant if it is invariant and if $g(z_1) = g(z_2)$ implies $z_1 = \mathbf M z_2$ for some $\mathbf M \in \mathcal M$ \citep[Definition 2.4]{Eaton89}. In other words, a maximal invariant has the same value for all members of an orbit, and a different value for each different orbit. Any invariant function is a function of a maximal invariant \citep[Theorem 2.3]{Eaton89}. Thus, we need to show that $\mathbf r$ is a maximal invariant.

First, we show that $\mathbf r$ is invariant under scaling. Denote by $m_i$ the $i$th diagonal element of $\mathbf M$. The $i$th mean proportional growth rate for some fixed time interval $\Delta t$ is 
\begin{equation*}
r_i = \frac{1}{\Delta t} \log \left( \frac{y_i}{x_i} \right) = \frac{1}{\Delta t} \log \left( \frac{m_i y_i}{m_i x_i} \right).
\end{equation*}
Thus the vector $\mathbf r$ of mean proportional growth rates is scaling invariant.

Second, we show that $\mathbf r^{(1)} = \mathbf r^{(2)}$ implies $\mathbf r^{(1)} = \mathbf M \mathbf r^{(2)}$, for two mean proportional growth rate vectors $\mathbf r^{(1)}$, $\mathbf r^{(2)}$, some matrix $\mathbf M \in \mathcal M$, and some fixed time interval $\mathbf \Delta t$. Consider the $i$th elements of $\mathbf r^{(1)}$ and $\mathbf r^{(2)}$:
\begin{equation*}
\begin{aligned}
r_i^{(1)} &= r_i^{(2)}\\
\frac{1}{\Delta t} \log \left( \frac{y^{(1)}_i}{x^{(1)}_i} \right) &=
 \frac{1}{\Delta t} \log \left( \frac{y^{(2)}_i}{x^{(2)}_i} \right) \\
\frac{y^{(1)}_i}{y^{(2)}_i} &= \frac{x^{(1)}_i}{x^{(2)}_i}\\
y^{(1)}_i &= \frac{x^{(1)}_i}{x^{(2)}_i} y^{(2)}_i\\
\end{aligned}
\end{equation*}
Thus, $z_1 = \mathbf M z_2$ for the diagonal matrix $\mathbf M$ with $i$th diagonal element $m_i = x^{(1)}_i / x^{(2)}_i$. In consequence, $\mathbf r$ is a maximal invariant under scaling, and any scaling invariant function must be a function of $\mathbf r$.
\end{proof}

In this paper, I considered only the mean and standard deviation of the elements of $\mathbf r$. By Theorem \ref{theorem:maxinvar} and the obvious result that any function of a maximal invariant is an invariant function, there are lots of other potentially interesting scaling invariant functions of $\mathbf r$. For example, contrasts (differences between the means of subsets of the elements of $\mathbf r$) are also scaling invariant, and can be used to measure differences in rates of change between groups of organisms. However, Theorem \ref{theorem:maxinvar} shows that there is no need to consider functions that cannot be written in terms of proportional growth rates. For example, a scaling invariant function cannot involve any kind of weighting by abundance, because such weightings do not appear in the maximal invariant $\mathbf r$.

\clearpage

\section{Estimating size and shape change}
\label{ap:estimate}

To estimate $s_r$ and $\bar{r}$ from the hoverfly data, estimates of abundance in each year are needed. The following simple approach is adequate to demonstrate the behaviour of $s_r$ and $\bar{r}$ for these data, and accounts for both parameter uncertainty and stochasticity.

I treated each species as independent. At each time $t$, I modelled the observed count $c_i(t)$ of the $i$th species as being generated by a Poisson distribution with parameter $\lambda_i(t)$, the expected abundance. I modelled the log of the Poisson parameter as a random walk (in discrete time, because the observations were made at annual intervals). The model is
\begin{equation*}
\begin{aligned}
\log \lambda_i(t+1) &= \log \lambda_i(t) + \epsilon_i(t) \\
\epsilon_i(t) &\sim \mathcal N(0,\sigma^2) \\
c_i (t) & \sim \text{Poisson}(\lambda_i(t))
\end{aligned}
\end{equation*}
This is about the simplest model that could be used to describe these data. In particular, it does not include density dependence, interspecific competition or consistent trends in proportional growth rates. It is not likely to make good long-term predictions, but it is being used for smoothing, rather than forecasting. Expected log abundances were almost indistinguishable from observed abundances.

I fitted this model in 64-bit R version 3.0.1 for Linux \citep{Rcore12} using the package \verb+sspir 0.2.10+ \citep{Dethlefsen06}. I used Iterated Extended Kalman Smoothing \citep{DK01} to obtain an approximating Gaussian model. I estimated $\sigma^2$ using maximum likelihood, with the Brent algorithm in R function \verb+optim+. I initialized the Kalman filter with $E(\log \lambda_i(0))= \log c_i(0)$ and $V(\log \lambda_i(0))= 1 \times 10^6$, so that expected initial abundance matched the observations but had high uncertainty. I did not include the first time interval in the analyses of patterns in $s_r$ and $\bar{r}$. This is a quick and dirty approach, but the Kalman filter converges quickly, so that the choice of initialization has little effect on the results for later time intervals.

Using the fitted approximating Gaussian model, I sampled 1000 replicate time series of the conditional distribution of $\log \lambda_i$ and used them to calculate the distributions of $s_r$ and $\bar{r}$ as follows:
\begin{enumerate}
\item Draw a value $\sigma^{2*}$ from the normal sampling distribution of $\sigma^2$, specified by the maximum likelihood estimate and its asymptotic variance.
\item Draw a value $\log \lambda_i^*(t)$ for each time from the conditional distribution of $\log \lambda_i(t)$, given the entire observation vector $\mathbf c_i = [c_i(0),c_i(1),\ldots]^{\prime}$. This was done using the \verb+ksimulate+ function in \verb+sspir+, with the sampled $\sigma^{2*}$.
\item Obtain simulated proportional growth rates
\begin{equation*}
\tilde{r}_i^*(t,t+1) = \log \lambda_i^*(t+1) - \log \lambda_i^*(t),
\end{equation*}
for intervals of length one year (Equation \ref{eq:meanrate} with $\Delta t = 1$).
\item Use the simulated proportional growth rates to calculate simulated values of $s_r$ (Equation \ref{eq:sdgrowth}) and $\bar{r}$ (Equation \ref{eq:lpi}).
\end{enumerate}
I report the mean and 0.025- and 0.975-quantiles of the simulated distributions of $s_r $ and $\bar{r}$. The whole procedure took less than an hour with an 3.2 GHz Intel Xeon processor and 16G RAM. The code is included as an electronic enhancement.

\clearpage
\section{Measuring colonization and extinction}
\label{ap:ce}

\subsection*{Motivation}

The proportional growth rate of a single species, $\frac{1}{x} \frac{dx}{dt}$, is undefined if $x=0$. This makes it difficult to deal with colonization and extinction in the framework of proportional change. Here, I describe a solution.

When calculating a mean proportional growth rate (Equation \ref{eq:meanrate}) it is reasonable to replace zero abundances $x=0$ by $\lim_{x \rightarrow 0^+} x$, since abundances cannot be negative. Then because $\lim_{x \rightarrow 0^+} \log x = -\infty$, a colonization or extinction can be thought of as representing an infinite amount of proportional change. This creates additional difficulties. First, it does not discriminate between extinction of rare and common species, because $-\infty + \log x(t) = -\infty$ for any real $ \log x(t)$. Second, it does not discriminate between fast and slow extinctions, because $-\infty/\Delta t = -\infty$ for any real $\Delta t$. Third, the arithmetic operations involved in the size and shape components of change (Equations \ref{eq:rbar} and \ref{eq:sdgrowth}) will often be undefined.

All these problems can be solved by working with surreal rather than real numbers. The surreal numbers \citep{Conway01} are a field which contains all the real numbers, and many others besides, including infinite and infinitesimal numbers. Arithmetic operations such as addition and multiplication, and logical comparisons such as greater than, are well-defined for all surreal numbers. Instead of taking the limit as $x \rightarrow 0^+$, zero abundances can be replaced by the surreal number $1/\omega$, which is greater than zero but smaller than all positive real numbers. This is similar to the idea of non-standard analysis, in which infinitesimals are used instead of limits \citep[p. 44]{Conway01}. The integral in Equation \ref{eq:meanrate} is valid for surreal numbers \citep[section 5.1]{Fornasiero03}. The natural logarithm function is defined for all positive surreal numbers, and $\log (1/\omega) = -\omega^{1/\omega}$ \citep[pp. 161-163]{Gonshor86}, which is more negative than all the real numbers, and therefore satisfies the requirement that a colonization or extinction represents a greater amount of proportional change than any event not involving a colonization or extinction. The choice $1/\omega$ is not unique, because there are many other numbers greater than zero but smaller than all positive real numbers (for example, $2/\omega$ and $1/\omega^2$). Choosing one of these other numbers would give the same qualitative result, but $1/\omega$ is in some sense the simplest number having the required properties \citep[p. 12]{Conway01}.

\subsection*{Size change}

The among-species mean proportional growth rates (the size component) can be calculated as follows. For each time interval $(t, t+\Delta t]$, divide the set of species being considered into four subsets, depending on whether they are present or absent at the start and end of the interval:
\begin{equation*}
\begin{aligned}
\mathcal S_1 &= \{i: x_i(t) >0 \land x_i(t+\Delta t) > 0\}, \quad \text{(always present)},\\
\mathcal S_2 &= \{i: x_i(t) >0 \land x_i(t+\Delta t) = 0\}, \quad \text{(extinctions)},\\
\mathcal S_3 &= \{i: x_i(t) =0 \land x_i(t+\Delta t) > 0\}, \quad \text{(colonizations)},\\
\mathcal S_4 &= \{i: x_i(t) =0 \land x_i(t+\Delta t) = 0\}, \quad \text{(never present)}.
\end{aligned}
\end{equation*}
Denote by $k_1,k_2,k_3,k_4$ the cardinalities of these sets. Denote by $a_i$ the function
\begin{equation*}
a_i = \begin{cases} \log x_i(t+\Delta t), & x_i(t+\Delta t) >0, \\
-\omega^{1/\omega}, & x_i(t + \Delta t) =0,
\end{cases}
\end{equation*}
Denote by $b_i$ the similar function of $x_i(t)$. Denote by $a$ and $b$ the sums of log abundances for species present at times $t+\Delta t$ and $t$ respectively:
\begin{equation*}
\begin{aligned}
a &= \sum_{i \in \mathcal S_1 \cup S_3} a_i, \\
b &= \sum_{i \in \mathcal S_1 \cup S_2} b_i. \\
\end{aligned}
\end{equation*}
Then the among-species mean of mean proportional growth rates over $(t,t+\Delta t]$ is
\begin{equation}
\begin{aligned}
\bar{r} &= \frac{1}{n} \sum_{i=1}^n \frac{a_i - b_i}{\Delta t} \\
&= \frac{1}{n \Delta t} \bigg[ \sum_{i \in \mathcal S_1} (\log x_i(t+\Delta t) - \log x_i(t)) \\ 
&\quad + \sum_{i \in \mathcal S_2} (-\omega^{1/\omega} - \log x_i(t)) \\
&\quad + \sum_{i \in \mathcal S_3} (\log x_i(t + \Delta t) - (-\omega^{1/\omega})) \\
&\quad + \sum_{i \in \mathcal S_4} (-\omega^{1/\omega} - (-\omega^{1/\omega})) \bigg]\\ 
&= \frac{1}{n \Delta t} [(k_3-k_2)\omega^{1/\omega} + a -b  ].
\end{aligned}
\label{eq:infmean}
\end{equation}
This reduces to Equation \ref{eq:rbar} if all species have non-zero abundance at all times. Otherwise, $\bar{r}$ has an infinite part which will be positive if there are more colonizations than extinctions, zero if there are equal numbers of colonizations and extinctions, and negative if there are more extinctions than colonizations. The finite part of $\bar{r}$ has contributions from changes in abundance that do not involve colonization or extinction, the initial abundances of species that went extinct, and the final abundances of species that colonized.

\subsection*{Shape change}

To calculate the among-species standard deviation of mean proportional growth rates (the shape component), it is easiest to first find the among-species variance and then take its square root. The among-species sample variance is
\begin{equation}
\frac{1}{n-1}\left( \sum_{i=1}^n \tilde{r}_i^2 - n\bar{r}^2 \right).
\label{eq:samplevardef}
\end{equation}
The first term in the parentheses in Equation \ref{eq:samplevardef} is
\begin{equation}
\begin{aligned}
\sum_{i=1}^n \tilde{r}_i^2 &= \frac{1}{(\Delta t)^2} \sum_{i=1}^n (a_i - b_i)^2 \\
&=\frac{1}{(\Delta t)^2} \bigg[ \sum_{i \in \mathcal S_1}(\log x_i(t+\Delta t) - \log x_i(t))^2\\
&\quad + \sum_{i \in \mathcal S_2}(-\omega^{1/\omega}-\log x_i(t))^2 \\
&\quad + \sum_{i \in \mathcal S_3}(\log x_i(t)-(-\omega^{1/\omega}))^2 \\
&\quad + \sum_{i \in \mathcal S_4}(-\omega^{1/\omega}-(-\omega^{1/\omega}))^2  \bigg] \\
&= \frac{1}{(\Delta t)^2} \bigg[(k_2 + k_3) \omega^{2/\omega} + 2 \omega^{1/\omega} \left( \sum_{i \in \mathcal S_2} \log x_i(t) + \sum_{i \in \mathcal S_3} \log x_i(t+\Delta t) \right)\\
&\quad + \sum_{i \in \mathcal S_1}(\log x_i(t+\Delta t) - \log x_i(t))^2 + \sum_{i \in \mathcal S_2} (\log x_i(t))^2 + \sum_{i \in \mathcal S_3} (\log x_i(t + \Delta t))^2 \bigg].
\end{aligned}
\label{eq:rtilde2}
\end{equation}
From Equation \ref{eq:infmean}, the second term in the parentheses in Equation \ref{eq:samplevardef} is
\begin{equation}
\begin{aligned}
n \bar{r}^2 &= n \left[ \frac{1}{n \Delta t} [(k_3-k_2)\omega^{1/\omega} + a - b]\right]^2 \\
&= \frac{1}{n (\Delta t)^2} [(k_3-k_2)^2 \omega^{2/\omega} + 2(a-b)(k_3-k_2)\omega^{1/\omega} + (a -b)^2]. \\
\end{aligned}
\label{eq:srbar2}
\end{equation}
Using Equations \ref{eq:rtilde2} and \ref{eq:srbar2}, the among-species sample variance in mean proportional growth rates can be written in the form
\begin{equation}
\frac{1}{(n-1)(\Delta t)^2} (\alpha \omega^{2/\omega} + \beta \omega^{1/\omega} + \gamma),
\label{eq:samplevar}
\end{equation}
where the coefficients $\alpha, \beta, \gamma$ are the real numbers
\begin{equation*}
\begin{aligned}
\alpha &= k_2 + k_3 - \frac{1}{n}(k_3-k_2)^2,\\
\beta &= 2\left( \sum_{i \in \mathcal S_2} \log x_i(t) + \sum_{i \in \mathcal S_3} \log x_i(t+ \Delta t) - \frac{1}{n}(a-b)(k_3-k_2) \right),\\
\gamma &= \left(\sum_{i \in \mathcal S_1} (\log x_i(t + \Delta t) - \log x_i(t))^2 + \sum_{i \in \mathcal S_2} (\log x_i(t))^2 + \sum_{i \in \mathcal S_3}(\log x_i(t+\Delta t)^2 -\frac{1}{n} (a-b)^2\right).
\end{aligned}
\end{equation*}
If all species have non-zero abundance at all times, $\alpha=\beta=0$ and the sample standard deviation is given by Equation \ref{eq:sdgrowth}. Otherwise, to find the sample standard deviation, write the second factor in Equation \ref{eq:samplevar} as
\begin{equation*}
\alpha \omega^{2/\omega} + \beta \omega^{1/\omega} + \gamma = \alpha \omega^{2/\omega}(1 + \delta),
\end{equation*}
where
\begin{equation*}
\delta = \frac{\beta \omega^{1/\omega} + \gamma}{\alpha \omega^{2/\omega}}
\end{equation*}
is an infinitesimal number. Then
\begin{equation}
[\alpha \omega^{2/\omega}(1 + \delta)]^{1/2} = \alpha^{1/2} \omega^{1/\omega} \left[1 + \frac{1}{2}\delta + \frac{1}{2}\left(\frac{1}{2} -1 \right) \frac{\delta^2}{2} + \ldots \right] \\
\label{eq:ConwayTaylor}
\end{equation}
\citep[Theorem 24]{Conway01}. Using Equations \ref{eq:samplevar} and \ref{eq:ConwayTaylor}, the sample standard deviation is
\begin{equation*}
\frac{1}{(n-1)^{1/2}\Delta t}\left( \alpha^{1/2} \omega^{1/\omega} + \frac{\beta}{2 \alpha^{1/2}} + \text{infinitesimal terms} \right).
\end{equation*}
Unless the infinitesimal terms can be shown to have any biological interpretation, it seems reasonable to discard them, leaving just the infinite and real terms (at least for drawing graphs).

\subsection*{Activity and $\cos^2 \theta$}

The norm $\| \mathbf r \|$ is given by the square root of Equation \ref{eq:rtilde2}, which can be evaluated using Equation \ref{eq:ConwayTaylor}. Values of the scaled activity level $n^{-1/2} \| \mathbf r \|$ can be ranked by comparing coefficients of powers of $\omega$ in Equation \ref{eq:rtilde2} in descending order. The value of the first coefficient that differs between two values determines their rank.

The squared norm $\| \mathbf u \|^2$ is given by Equation \ref{eq:srbar2}. The proportion of change that is size change, $\cos^2 \theta = (\| \mathbf u \| / \| \mathbf r \|)^2$, will then be a ratio of the form
\begin{equation}
\frac{a \omega^{2/\omega} + b \omega ^{1/\omega} + c}{d \omega^{2/\omega} + e \omega ^{1/\omega} + f},
\label{eq:cos2ratio}
\end{equation}
with real coefficients $a, b, c, d, e, f$. Dividing the numerator and the denominator by $\omega^{2/\omega}$,
\begin{equation*}
\cos^2 \theta = \frac{a}{d + e \omega ^{-1/\omega} + f \omega^{-2/\omega}}  + \text{infinitesimal terms},
\end{equation*}
whose value will be close to $a/d$. I therefore use the real approximation $\cos^2 \theta \simeq a/d$ in graphs. However, the infinitesimal parts may be important in resolving ties. It is straightforward to rank two ratios of the form (\ref{eq:cos2ratio}), by cross-multiplying by their denominators and comparing the coefficients of matching powers of $\omega$. I use this approach when calculating statistics based on the ranks of $\cos^2 \theta$.

\clearpage
\section{Size and shape change in presence-absence data}
\label{ap:pa}

In the notation of \ref{ap:ce}, the mean proportional growth rate over the interval $(t, t + \Delta t]$ for the $i$th species is $\tilde{r}_i = (a_i - b_i)/\Delta t$. The coefficient of $\psi = \omega^{1/\omega}$ in $\tilde{r}_i$ is the only part that can be calculated if only presence/absence data are available. Denote this coefficient by $r^{(\psi)}_i$. Its value is
\begin{equation}
r^{(\psi)}_i = \begin{cases}
\frac{-1}{\Delta t}, & \text{extinction},\\
\frac{1}{\Delta t}, & \text{colonization},\\
0, & \text{otherwise}.
\end{cases}
\label{eq:pacode}
\end{equation}
The among-species mean of these coefficients is a natural measure of size change:
\begin{equation}
\begin{aligned}
\frac{1}{n \Delta t} \sum_{i=1}^n r^{(\psi)}_i &= k_2(-1) + k_3(1)\\
&= \frac{k_3 - k_2}{n \Delta t},
\end{aligned}
\label{eq:sizepa}
\end{equation}
where $k_2$ is the number of extinctions and $k_3$ is the number of colonizations. Equation \ref{eq:sizepa} is the coefficient of $\psi$ in Equation \ref{eq:sizeinf}.

The among-species sample standard deviation of these coefficients is a natural measure of shape change:
\begin{equation*}
\begin{aligned}
\frac{1}{\Delta t}\left[ \frac{1}{n-1} \left( \sum_{i=1}^n (r^{(\psi)}_i)^2 - \frac{1}{n} \left(\sum_{i=1}^n r^{(\psi)}_i \right)^2 \right) \right]^{1/2} \\
= \frac{1}{\Delta t}\left[ \frac{1}{n-1} \left( k_2(-1)^2 + k_3(1)^2 - \frac{1}{n} (k_3(1) + k_2(-1))^2 \right) \right]^{1/2} \\
=\frac{1}{\Delta t}\left[ \frac{1}{n-1} \left( k_2 + k_3 - \frac{1}{n} (k_3 - k_2)^2 \right) \right]^{1/2}, \\
\end{aligned}
\label{eq:shapepa}
\end{equation*}
which is the coefficient of $\psi$ in Equation \ref{eq:sdinf}.

The coefficient of $\psi^2$ in the squared activity level is
\begin{equation*}
\begin{aligned}
\sum_{i=1}^n (r^{(\psi)}_i)^2 &= \frac{k_2(-1)^2 + k_3(1)^2}{(\Delta t)^2}\\
&=\frac{k_2 + k_3}{(\Delta t)^2}.
\end{aligned}
\end{equation*}
Thus, the measures of rate of shape change proposed by \citet{Anderson07}, which are proportional to $k_2 + k_3$, are measures of squared activity.

Finally, using the results above, the scaled activity from presence-absence data is $(1/( n ^{-1/2}\Delta t)) (k_2 + k_3)^{1/2}$, and the proportion of change that is size change is $\cos^2 \theta = (k_3 -k_2)^2/(n(k_2 + k_3))$.

\end{appendices}

\clearpage
\section*{References}
%\bibliographystyle{elsarticle-harv}
%\bibliography{journal_fullnames,matt_database}

\end{document}